\documentclass[10pt]{article}

\usepackage[utf8]{inputenc}
\usepackage{float}

\usepackage{amsthm}
\usepackage[all]{xy}
\usepackage{amsmath}
\usepackage{amssymb}
\usepackage{amsfonts}
\usepackage{authblk}
\usepackage{enumerate}

\newenvironment{tbs}{%
   \small\tt
   \begin{itemize}}{\end{itemize}}
\newcommand{\btbs}{\begin{tbs}}                                                                      
\newcommand{\etbs}{\end{tbs}}

\newcommand{\hide}[1]{}

\renewcommand{\phi}{\varphi}

\newtheorem{theo}{Theorem}

\newtheorem{prop}{Proposition}
\newtheorem{lemma}{Lemma}
\newtheorem{coro}{Corollary}
\newtheorem{fact}{Fact}
\newtheorem{question}{Question}

\theoremstyle{definition}
\newtheorem{defi}{Definition}
\theoremstyle{definition}
\newtheorem{example}{Example}
\theoremstyle{definition}
\newtheorem{remark}{Remark}

\newcommand{\isdef}{\mathrel{:=}}

\newcommand{\bbA}{\mathbb{A}}
\newcommand{\bbS}{\mathbb{S}}
\newcommand{\bbT}{\mathbb{T}}

\newcommand{\fun}{\mathsf{T}}
\newcommand{\psf}{\mathcal{P}}
\newcommand{\cvp}{\mathcal{Q}}
\newcommand{\mon}{\mathcal{M}}
    \newcommand{\monstar}{\mathcal{M}^{\star}}

\newcommand{\mathlang}[1]{\mathtt{#1}}
   
\newcommand{\MSO}{\mathlang{MSO}}
   \newcommand{\MSOT}{\MSO_{\fun}}
   \newcommand{\MSOLa}{\MSO_{\Lambda}}
\newcommand{\muML}{\mathlang{\mu ML}}
   \newcommand{\muMLT}{\muML_{\fun}}
   \newcommand{\muMLLa}{\muML_{\Lambda}}
\newcommand{\ML}{\mathlang{ML}}
   \newcommand{\MLT}{\mathlang{ML}_{\fun}}
   \newcommand{\MLLa}{\mathlang{ML}_{\Lambda}}
\newcommand{\SO}{\mathlang{SO}}
   \newcommand{\SOT}{\mathlang{SO}_{\fun}}
   \newcommand{\SOLa}{\mathlang{SO}_{\La}}

\newcommand{\Var}{\mathit{Var}}

\newcommand{\isbnf}{\mathrel{::=}}
\newcommand{\divbnf}{\mathrel{|}}

\newcommand{\sr}{\mathtt{sr}}
\newcommand{\Em}{\mathtt{Em}}

\newcommand{\Sing}{\mathtt{Sing}}

\newcommand{\Aut}{\mathit{Aut}}

\newcommand{\smod}{\mathbb{S}}

\newcommand{\si}{\sigma}

\newcommand{\De}{\Delta}
\newcommand{\La}{\Lambda}
\newcommand{\Om}{\Omega}

\title{Monadic Second-Order Logic and Bisimulation Invariance for Coalgebras}

\author[1,2]{Sebastian Enqvist}
\author[1]{Fatemeh Seifan}
\author[1]{Yde Venema}
\affil[1]{Institute for Logic, Language and Computation,  University of Amsterdam}
\affil[2]{Department of Philosophy, Lund University}

\begin{document}
\maketitle
\begin{abstract}
Generalizing standard monadic second-order logic for Kripke models, we 
introduce monadic second-order logic $\MSO_{\fun}$ interpreted
over coalgebras for an arbitrary set functor $\fun$. 
Similar to well-known results for monadic second-order logic over trees, we 
provide a translation of this logic into a class of automata, relative to the class 
of  $\fun$-coalgebras that admit a tree-like supporting Kripke frame. 
We then consider invariance under behavioral equivalence of 
$\MSO_{\fun}$-formulas; more in particular, we investigate whether the coalgebraic 
$\mu$-calculus is the bisimulation-invariant fragment of $\MSO_{\fun}$.
Building on recent results by the third author we show that in order to provide
such a coalgebraic generalization of the Janin-Walukiewicz Theorem, it suffices 
to find what we call an \textit{adequate uniform construction} for the functor $\fun$.
As applications of this result we obtain a partly new proof of the Janin-Walukiewicz 
Theorem, and bisimulation invariance results for the bag functor (graded modal
logic) and all exponential polynomial functors.

Finally, we consider in some detail the monotone neighborhood functor $\mon$, 
which provides coalgebraic semantics for monotone modal logic. 
It turns out that  there is \textit{no} adequate uniform construction for $\mon$, 
whence the automata-theoretic approach towards bisimulation invariance
does not apply  directly. 
This problem can be overcome if we consider \textit{global} bisimulations between
neighborhood models: one of our main technical results provides a characterization of the
monotone modal $\mu$-calculus extended with the \textit{global modalities}, as the
fragment of monadic second-order logic for the monotone neighborhood functor
that is invariant for global bisimulations. 
\end{abstract}

\section{Introduction}

\subsection{Logic, automata and coalgebra}
The aim of this paper is to strengthen the link between the areas of logic, 
automata and coalgebra.
More in particular, we provide a coalgebraic generalization of the
automata-theoretic approach towards monadic second-order logic ($\MSO$), and
we address the question whether the Janin-Walukiewicz Theorem can be 
generalized from Kripke structures 
to the setting of arbitrary coalgebras.

The connection between \textit{monadic second-order logic} and \textit{automata}
is classic, going back to the seminal work of B\"uchi, Rabin, and others.
For instance, Rabin's decidability result for the monadic second-order theory
of binary trees, or $S2 S$, makes use of a translation of 
monadic second-order logic into a class of automata, thus reducing the 
satisfiability problem for $S 2 S$ to the non-emptiness problem for the
corresponding automata \cite{rabi:deci69}.
The link between $\MSO$ and automata over trees with arbitrary branching
was further explored by Walukiewicz~\cite{walu:mona96}.
Janin and Walukiewicz considered monadic second-order logic interpreted over 
Kripke structures, and used automata-theoretic techniques to obtain a van 
Benthem-like characterization theorem for monadic second-order logic, identifying
the modal $\mu$-calculus as the bisimulation invariant fragment of 
$\MSO$ \cite{jani:expr96}.
Given the fact that in many applications bisimilar models are considered
to represent the \emph{same} process, one has little interest in properties of 
models that are \emph{not} bisimulation invariant.
Thus the Janin-Walukiewicz Theorem can be seen as an expressive completeness 
result, stating that all \emph{relevant} properties in monadic second-order 
logic can be expressed in the modal $\mu$-calculus.

Coalgebra enters naturally into this picture.
Recall that Universal Coalgebra~\cite{rutt:univ00} provides the notion of a 
\emph{coalgebra} as the natural mathematical generalization of state-based 
evolving systems such as streams, (infinite) trees, Kripke models, 
(probabilistic) transition systems, and many others.
This approach combines simplicity with generality and wide applicability: many
features, including input, output, nondeterminism, probability, and interaction, 
can easily be encoded in the coalgebra type $\fun$ (formally an endofunctor on 
the category $\mathbf{Set}$ of sets as objects with functions as arrows).
Starting with Moss' seminal paper~\cite{moss:coal99}, coalgebraic logics have 
been developed for the purpose of specifying and reasoning about \emph{behavior},
one of the most fundamental concepts that allows for a natural coalgebraic 
formalization.
And with Kripke structures constituting key examples of coalgebras, it should
come as no surprise that most coalgebraic logics are some kind of modification
or generalization of \emph{modal logic}~\cite{cirs:moda11}.

The coalgebraic modal logics that we consider here originate with
Pattinson~\cite{patt:coal03}; they are characterized by a completely standard
syntax, in which the semantics of each modality is determined by a so-called 
\emph{predicate lifting} (see Definition~\ref{d:pl} below).
Many well-known variations of modal logic in fact arise as the coalgebraic logic
$\ML_{\La}$ associated with a set $\La$ of such predicate liftings; examples
include both standard and (monotone) neighborhood modal logic, graded and 
probabilistic modal logic, coalition logic, and conditional logic.
Extensions of coalgebraic modal logics with fixpoint operators, needed 
for describing \emph{ongoing} behavior, were developed 
in~\cite{vene:auto06,cirs:expt09}.

The link between coalgebra and automata theory is by now well-established.
For instance, finite state automata operating on finite words have been 
recognized as key examples of coalgebra from the 
outset~\cite{rutt:univ00}.
More relevant for the purpose of this paper is the link with precisely the
kind of automata mentioned earlier, since the (potentially infinite) objects 
on which these devices operate, such as streams, trees and Kripke frames,
usually are coalgebras.
Thus, the automata-theoretic perspective on modal fixpoint logic could be
lifted to the abstraction level of coalgebra~\cite{vene:auto06,font:auto10}.
In fact, many key results in the theory of automata operating on infinite 
objects, such as Muller \& Schupp's Simulation Theorem~\cite{mull:simu95} can
in fact be seen as instances of more general theorems in Universal 
Coalgebra~\cite{kupk:coal08}.


\subsection{Coalgebraic monadic second-order logic}

Missing from this picture is, to start with, a coalgebraic version of 
\emph{(monadic) second-order logic}.
Filling this gap is the first aim of the current paper, which introduces a notion
of \emph{monadic second-order logic} $\MSOT$ for coalgebras of type $\fun$.
Our formalism combines two ideas from the literature.
First of all, we looked for inspiration to the coalgebraic versions of 
\emph{first-order logic} of Litak \& alii~\cite{lita:coal12}.
These authors introduced Coalgebraic Predicate Logic as a common generalisation 
of first-order logic and coalgebraic modal logic, combining first-order 
quantification with coalgebraic syntax based on predicate liftings.
Our formalism $\MSOT$ will combine a similar syntactic feature with second-order
quantification.
Second, following the tradition in automata-theoretic approaches towards monadic
second-order logic, our formalism will be \emph{one-sorted}.
That is, we \emph{only} allow second-order quantification in our language, 
relying on the fact that individual quantification, when called for, can be
encoded as second-order quantification relativized to singleton sets.
Since predicate liftings are defined as families of maps on powerset
algebras, these two ideas fit together very well, to the effect that our
second-order logic is in some sense simpler than the first-order 
formalism of~\cite{lita:coal12}.

In section~\ref{sec:mso} we will define, for any set $\Lambda$ of 
monotone\footnote{%
  In the most general case, restricting to monotone predicate liftings is not 
  needed, one could define $\MSOT$ as the logic obtained by taking for $\La$
  the set of \emph{all} predicate liftings. 
  However, in the context of this paper, where we take an automata-theoretic
  perspective on $\MSO$, this restriction makes sense.
  }
predicate liftings, a formalism $\MSOLa$, and we let $\MSOT$ denote the logic
obtained by taking for $\Lambda$ the set of \emph{all} monotone predicate 
liftings.
Clearly we will make sure that this definition generalizes the standard case,
in the sense that the standard version of $\MSO$ for Kripke structures 
instantiates the logic $\MSO_{\{\Diamond\}}$ and is equivalent to the coalgebraic
logic $\MSO_{\psf}$ (where $\psf$ denotes the power set functor).

The introduction of a monadic second-order logic $\MSOT$ for $\fun$-coalgebras 
naturally raises the question, for which $\fun$ the coalgebraic modal
$\mu$-calculus for $\fun$ corresponds to the bisimulation-invariant fragment of 
$\MSOT$.

\begin{question}
\label{q:Q}
\text{Which functors $\fun$ satisfy $\muMLT \equiv \MSOT/{\simeq}$?}
\end{question}


\subsection{Automata for coalgebraic monadic second-order logic}

In order to address Question~\ref{q:Q}, we take an automata-theoretic 
perspective on the logics $\MSOT$ and $\muMLT$, and as the second contribution
of this paper we introduce a class of parity automata for $\MSOT$.

As usual, the operational semantics of our automata is given in terms
of a two-player acceptance game, which proceeds in \emph{rounds} moving from
one basic position to another, where a basic position is a pair consisting of
a state of the automaton and a point in the coalgebra structure under 
consideration.
In each round, the two players, $\exists$ and $\forall$, focus on a certain
local `window' on the coalgebra structure.
This `window' takes the shape of a \emph{one-step $\fun$-model}, that is, a 
triple $(X,\alpha,V)$ consisting of a set $X$, a \emph{chosen object} $\alpha
\in \fun X$, and 
a valuation $V$ interpreting the states of the automaton as subsets of $X$.
More specifically, during each round of the game it is the task of $\exists$
to come up with a valuation $V$ that creates a one-step model in which a 
certain \emph{one-step formula} $\delta$ (determined by the current basic 
position in the game) is true.

Generally, our automata will have the shape $\bbA = (A,\De,\Om,a_{I})$
where $A$ is a finite carrier set with initial state $a_{I} \in A$, and $\Om$
and $\De$ are the parity and transition map of $\bbA$, respectively.
The flavour of such an automaton is largely determined by the co-domain of its
transition map $\De$, the so-called \emph{one-step language} which consists of 
the one-step formulas that feature in the acceptance game as 
described.

Each one-step language $\mathlang{L}$ induces its own class of automata 
$\Aut(\mathlang{L})$.
For instance, the class of automata corresponding to the coalgebraic fixpoint
logic $\muMLLa$ can be given as $\Aut(\ML_{\La})$, where $\ML_{\La}$ is the set
of positive modal formulas of depth one that use modalities from 
$\La$~\cite{font:auto10}.
Basically then, the problem of finding the right class of automata for the
coalgebraic monadic second-order logic $\MSOLa$ consists in the identification
of an appropriate one-step language.
Our proposal comprises a one-step \emph{second-order logic} which uses predicate 
liftings to describe the chosen object of the one-step model.

Finally, note that similar to the case of standard $\MSO$, the equivalence 
between formulas in $\MSOT$ and automata in $\Aut(\SO)$ is only guaranteed to 
hold for coalgebras that are `tree-like' in some sense (to be defined further 
on).

\begin{theo}[Automata for coalgebraic $\MSO$]
\label{t:automatachar}
For any set $\La$ of monotone predicate liftings for $\fun$ there is an
effective construction mapping any formula $\varphi\in \MSOLa$ into an
automaton $\mathbb{A}_\varphi \in \Aut(\SOLa)$, 
which is
equivalent to $\varphi$ over $\fun$-tree models.
\end{theo}

The proof of Theorem~\ref{t:automatachar} proceeds by induction on the 
complexity of $\MSOT$-formulas, and thus involves various \emph{closure 
properties} of automata, such as closure under complementation, union and
projection.
In order to establish these results, it will be convenient to take an
\emph{abstract} perspective, revealing how closure properties of a class of
automata are completely determined at the level of the one-step language.

\subsection{Bisimulation Invariance}

With automata-theoretic characterizations in place for both coalgebraic $\MSO$ 
and the coalgebraic $\mu$-calculus $\muML$, we can address Question~\ref{q:Q}
by considering the following question:
\begin{question}
\label{q:Qa}
\text{Which functors $\fun$ satisfy $\Aut(\ML)\equiv \Aut(\SO)/{\simeq}$?}
\end{question}
Continuing the program of the third author~\cite{vene:expr14}, we will approach
this question \emph{at the level of the one-step languages}, $\SO$ and $\ML$.
To start with, observe that any translation (from one-step formulas in) $\SO$ to
(one-step formulas in) $\ML$ naturally induces a translation from $\SO$-automata
to $\ML$-automata.
A new observation we make here is that any so-called \emph{uniform construction}
on the class of one-step models for the functor $\fun$ that satisfies certain
\emph{adequacy} conditions, provides
(1) a translation $(\cdot)^{*}: \SO \to \ML$, together with
(2) a construction $(\cdot)_{*}$ transforming a pointed $\fun$-model $(\bbS,s)$
into a tree model $(\bbS_{*},s_{*})$ which is a coalgebraic pre-image of 
$(\bbS,s)$ satisfying
\[
\bbA \text{ accepts } (\bbS_{*},s_{*}) \text{ iff }
\bbA^{*} \text{ accepts } (\bbS,s).
\]
From this it easily follows that an $\SO$-automaton $\bbA$ is bisimulation 
invariant iff it is equivalent to the $\ML$-automaton $\bbA^{*}$.

On the basis of these observations we can prove the following generalisation of 
the Janin-Walukiewicz Theorem.

\begin{theo}[Coalgebraic Bisimulation Invariance]
\label{t:main2}
If the set functor $\fun$ admits an adequate uniform construction, then 
\[
\muMLT \equiv \MSOT/{\simeq}.
\]
\end{theo}

In our eyes, the significance of Theorem~\ref{t:main2} is twofold.
First of all, the proof separates the `clean', abstract part of 
bisimulation-invariance results from the more functor-specific parts.
As a consequence, Theorem~\ref{t:main2} can be used to obtain immediate results
in particular cases.
Examples include the power set functor (standard Kripke structures),
where the adequate uniform construction roughly consists of taking $\omega$-fold 
products (see Example~\ref{ex:psf}), the bag functor (Example~\ref{ex:bag}),
and all exponential polynomial functors (Corollary~\ref{c:epf}).
Second, in case the functor does \emph{not} admit an adequate uniform
construction, Theorem~\ref{t:main2} may still be of use in proving alternative 
characterization results for the functor.

Instantiating the latter phenomenon is the \emph{monotone neighborhood functor} 
$\mon$ (see the next section for its definition).
The importance of this functor lies, among other things, in it providing a 
coalgebraic semantics for monotone modal logic~\cite{hans:coal04}. 
The coalgebraic monadic second-order language $\MSO_\mon$  is equivalent to a 
natural second-order language for reasoning about monotone neighborhood 
structures that we shall denote by $\mathtt{MMSO}$, and $\muML_\mon$ is 
equivalent to the fixpoint-extension of the monotone $\mu$-calculus, denoted 
$\mu \mathtt{MML}$.
As we shall see in Proposition~\ref{p:no-adc-mon} below, $\mon$ does \emph{not}
admit an adequate uniform construction.\footnote{%
   This does not necessarily mean that the monotone $\mu$-calculus $\muML_{\mon}$ 
   does \emph{not} correspond to the bisimulation-invariant fragment of 
   $\MSO_{\mon}$, but it does mean that a proof of such a result, if provable at 
   all, will be significantly more involved than for those functors where 
   Theorem~\ref{t:main2} does apply.}
This, however, is not the end of the story. 
It turns out that we \emph{can} find an adequate uniform construction for 
a \emph{variant} $\monstar$ of the functor $\mon$ (see 
Proposition~\ref{p:ad-monstar}).
As a corollary, we obtain a characterization of the fragment of 
$\mathtt{MMSO}$ that is invariant under \emph{global} bisimulations
(bisimulations that are full on both domain and codomain).
This fragment turns out to be exactly the extension of the monotone $\mu$-calculus
with the global modalities (for precise definitions we refer to 
section~\ref{sec:mon}), which we shall denote $\mu \mathtt{MML}_g$.

In this notation, our final contribution is the following characterization result:
\begin{theo}
\label{t:JWmonglob}
A formula in $\mathtt{MMSO}$ is invariant for global neighborhood 
bisimulations if, and only if, it is equivalent to a formula of the logic  
$\mu \mathtt{MML}_g$.
\end{theo}

\section{Some technical background}

In this paper we assume familiarity with the basic theory of modal (fixpoint) 
logic, monadic second-order logic, coalgebra, coalgebraic modal (fixpoint) 
logic, and parity games.
Here we fix some notation and terminology.

\subsection{Kripke models and their logics}

We restrict to the theory of modal logic with one modality (and hence, one 
accessibility relation).
Let $\Var$ be  a fixed infinite supply of variables. 
A \textit{Kripke model} is a structure $\bbS = (S,R,V)$ where $S$ is a set,
$R \subseteq S \times S$ and $V : \Var \rightarrow \psf(S)$ is a 
$\Var$-valuation.
Associated with such a valuation $V$, we define the \emph{conjugate coloring} 
$V^{\dagger}: S \to \psf(\Var)$ by $V^{\dagger}(s) \isdef \{ p \in \Var \mid
s \in V(p)\}$.
Given a subset $T \subseteq S$, the valuation $V[p \mapsto T]$ is as $V$ except 
that it maps the variable $p$ to $T$.
A \textit{pointed} Kripke model is a structure $(\bbS,u)$ where $\bbS$ is a
Kripke model and $u$ is a point in $\bbS$. 
Turning to syntax, we define the formulas of monadic second-order logic $\MSO$ 
through the following grammar:
$$
\varphi \isbnf \sr(p) \divbnf p \subseteq q \divbnf R(p,q)   
  \divbnf \neg \varphi \divbnf\varphi \vee\varphi \divbnf \exists p. \varphi,
$$
with $p,q \in \Var$. 
Formulas are evaluated over pointed Kripke models by the following induction:
\begin{itemize}
\item $(S,R,V,u) \vDash  \sr(p) $ iff $V(p) = \{u\}$
\item $(S,R,V,u)\vDash p \subseteq q $ iff $V(p) \subseteq V(q)$
\item $(S,R,V,u) \vDash R(p,q)$ iff for all $v \in V(p)$ there is $w \in V(q)$ 
  with $v R w$ 
\item standard clauses for the boolean connectives
\item $(S,R,V,u)\vDash \exists p. \varphi$ iff $(S,R,V[p \mapsto T],u)\vDash 
   \varphi$ for some $T \subseteq S$.
\end{itemize}

We present the language of the modal $\mu$-calculus $\muML$ in
negation normal form, by the following grammar:
$$ \varphi \isbnf p \divbnf \neg p \divbnf \bot \divbnf \top 
   \divbnf \phi \lor \phi \divbnf \phi \land \phi
   \divbnf \Box \varphi \divbnf \Diamond \varphi 
   \divbnf \eta p. \varphi 
$$
where $p \in \Var$, $\eta \in \{ \mu, \nu \}$,  and in the formula $\eta p. \varphi$
no free occurrence of the variable $p$ may be in the scope of a negation. 

The satisfaction relation between pointed Kripke models and formulas in 
$\muML$ is defined by the usual induction, with, e.g.
\begin{itemize}
\item 
$(S,R,V,u)\vDash \mu p. \varphi$ iff $u \in \bigcap \{Z \subseteq S 
\mid \varphi_{p}(Z) \subseteq Z\}$
where $\varphi_{p}(Z)$ denotes the truth set of the formula $\varphi$ in the 
model $(S,R,V[p \mapsto Z])$.
\end{itemize}

We assume familiarity with the notion of bisimilarity between two (pointed)
Kripke models, and say that a formula of $\MSO$ is \textit{bisimulation
invariant} if  it has the same truth value in any pair of bisimilar pointed
Kripke models. 

\begin{fact}\cite{jani:expr96}
A formula $\varphi$ of $\MSO$ is equivalent to a formula of $\muML$ iff
$\varphi$ is invariant for bisimulations.
\end{fact}

\subsection{Coalgebras and models}

Our basic semantic structures consist of coalgebras together with valuations.
We only consider coalgebras over the base category $\mathbf{Set}$ with sets as
objects and functions as arrows.
The co- and contravariant power set functors will be denoted by $\psf$ and 
$\cvp : \mathbf{Set} \to \mathbf{Set}^{op}$, respectively.
Covariant endofunctors on $\mathbf{Set}$ will be called \textit{set 
functors}.

\begin{defi}
Let $\fun$ be a set functor. 
A \textit{$\fun$-coalgebra} is a pair $(S,\si)$ consisting of a set $S$, 
together with a map $\sigma : S \to \fun S$.
A $\fun$-\textit{model} is a structure $\mathbb{S} = (S,\sigma,V)$ where 
$(S,\sigma)$ is a $\fun$-coalgebra and $V : \Var \rightarrow \psf S$.
A \textit{pointed} $\fun$-model is a structure $(\mathbb{S},s)$ where 
$\mathbb{S}$ is a $\fun$-model and $s \in S$.
\end{defi}

The usual notion of a $p$-morphism between Kripke models can be generalized as
follows:
Let $\mathbb{S}_1 = (S_1,\sigma_1,V_1)$ and $\mathbb{S}_2 = (S_2,\sigma_2,V_2)$
be two $\fun$-models and let $f : S_1 \rightarrow S_2$ be any map. 
Then $f$ is said to be a $\fun$-\textit{model homomorphism} if:
\begin{enumerate}
\item 
for each variable $p$ and each $u \in S_1$, we have $u \in V_1(p)$ iff 
$f(u) \in V_2(p)$;
\item 
the map $f$ is a \textit{coalgebra morphism}, i.e. we have
$$\sigma_2 \circ f = \fun f \circ \sigma_1.$$
\end{enumerate}
Two pointed coalgebras $(\bbS,s)$ and $(\bbS',s')$ are \emph{behaviorally
equivalent} if $s$ and $s'$ can be identified by coalgebra morphisms $f: 
\bbS \to \bbT$ and $f': \bbS' \to \bbT$ such that $f(s) = f'(s')$.

A \emph{coalgebraic logic} consists of a set $\mathlang{L}$ of \emph{formulas}
together with, for each coalgebra
$(S,\si)$, a truth or satisfaction relation ${\Vdash} \subseteq S \times
\mathlang{L}$.
A formula $\phi$ is called \emph{bisimulation invariant}\footnote{Strictly speaking, behavioral equivalence and bisimilarity are distinct concepts. However, in many concrete cases, behavioural equivalence and bisimilarity coincide, so we shall be content to use the more common parlance of ``bisimulation invariance'' rather than ``invariance for behavioural equivalence.}
if $\bbS, s 
\Vdash \phi \iff \bbS',s' \Vdash \phi$ whenever $\bbS,s \simeq \bbS', s'$.

Kripke frames are coalgebras for the (covariant) power set functor $\psf$.
A functor of particular interest in this paper is the \textit{monotone 
neighborhood} functor $\mon$, usually defined as the subfunctor of $\cvp \circ
\cvp$ given by
$$
\mon X = \{N \in \cvp \cvp X \mid \forall Z,Z': Z \in N 
\mathrel{\&} 
   Z \subseteq Z' \Rightarrow Z' \in N\}.$$
This functor comes equipped with the following notion of bisimilarity.
A \emph{neighborhood bisimulation} between $\mon$-models $\mathbb{S}_1$ and 
$\mathbb{S}_2$ is a relation $R \subseteq S_{1} \times S_{2}$ such that, if 
$s_1 R s_2$ 
then:
\begin{itemize}
\item $V_1^\dagger(s_1) = V_2^\dagger(s_2)$;
\item for all $Z_1$ in $\sigma_1(s_1)$ there is $Z_2 $ in $\sigma_2(s_2)$ such 
   that, for all $t_2 \in Z_2$ there is $t_1 \in Z_1$ with $t_1 R t_2$;
\item for all $Z_2$ in $\sigma_2(s_2)$ there is $Z_1 $ in $\sigma_1(s_1)$ such
   that, for all $t_1 \in Z_1$ there is $t_2 \in Z_2$ with $t_1 R t_2$.
\end{itemize}

\subsection{Coalgebraic $\mu$-calculus \& coalgebra automata}

The modal $\mu$-calculus is just one in a family of logical systems that may 
collectively be referred to as the \textit{coalgebraic 
$\mu$-calculus}~\cite{cirs:expt09}. 
These logics essentially make use of \textit{predicate liftings}. 

\begin{defi}
\label{d:pl}
Given a set functor $\fun$, an \textit{$n$-place predicate lifting} for $\fun$ 
is a natural transformation
$$\lambda : \cvp(-)^n \rightarrow \cvp \circ \fun,
$$
where $\cvp(-)^n$ denotes the $n$-fold product of $\cvp$ with itself.
A predicate lifting $\lambda$ is said to be \textit{monotone} if
$$
\lambda_X(Y_1,...,Y_n) \subseteq \lambda_X(Z_1,...,Z_n),
$$
whenever $Y_i \subseteq Z_i$ for each $i$. 
The \textit{Boolean dual}  $\lambda^d$ of $\lambda$ is defined by
$$(Z_1,...,Z_n) \mapsto \fun X \setminus 
   (\lambda_X(X\setminus Z_1,...,X \setminus Z_n)).
$$
\end{defi}

Given a set functor $\fun$, the language $\muMLT$ of the coalgebraic 
$\mu$-calculus for $\fun$ is defined thus:
$$
\varphi \isbnf p  \divbnf \neg p \divbnf \bot \divbnf  \top 
   \divbnf \lambda(\varphi_1,...,\varphi_n) 
   \divbnf \varphi \vee \varphi \divbnf  \varphi \wedge \varphi 
   \divbnf \eta p. \varphi
$$
where $p \in \Var$, $\lambda$ is any monotone $n$-place predicate lifting for 
$\fun$, $\eta \in \{\mu,\nu\}$, and, in $\eta p. \varphi$, no free occurrence of 
the variable $p$ is in the scope of a negation. 
If we restrict the formulas $\lambda(\varphi_1,...,\varphi_n)$ so that $\lambda$ 
must come frome some distinguished set of liftings $\Lambda$, then we denote the 
corresponding sublanguage of $\muMLT$ by $\muMLLa$. 

The semantics of formulas in a pointed $\fun$-model is defined as follows:
\begin{itemize}
\item 
$(\mathbb{S},s)\vDash p $ iff $s \in V(p)$ and $(\mathbb{S},s)\vDash \neg p$ iff 
   $s \notin V(p)$
\item $(\mathbb{S},s)\vDash \lambda(\varphi_1,...,\varphi_n)$ iff $\sigma(s)\in 
   \lambda_S(\Vert \varphi_1 \Vert,...,\Vert \varphi_n\Vert)$, where
  $\Vert \varphi_i \Vert = \{t \in S \mid (\mathbb{S},t)\vDash \varphi_i\}$
denotes the ``truth set'' of $\varphi_i$ in $\mathbb{S}$
\item standard clauses for the boolean connectives
\item $(\mathbb{S},s)\vDash \mu p. \varphi$ iff
  $s \in \bigcap \{X \subseteq S \mid \varphi_{p}(X) \subseteq X\}$,
  where $\varphi_{p}(Z)$ denotes the truth set of the formula $\varphi$ in the 
  $\fun$-model $(S,\si,V[p \mapsto Z])$.
\end{itemize}

It is routine to prove that all formulas in $\muMLT$ are bisimulation 
invariant.

Turning to the parity automata corresponding to the language $\muMLLa$, we
first define the \textit{modal one-step language} $\MLLa^{1}$.
Its set $\MLLa^{1}(A)$ of \emph{modal one-step formulas} over a set $A$ of
variables is given by the following grammar:
$$ 
\varphi \isbnf
   \bot \divbnf \top \divbnf \lambda(\psi_1,...,\psi_n) 
   \divbnf \varphi \vee \varphi \divbnf \varphi \wedge \varphi
$$
where $\psi_1,...,\psi_n$ are formulas built up from variables in $A$ using 
disjunctions and conjunctions.

\begin{defi}
Given a functor $\fun$ and a set of variables $A$, a \textit{one-step model} 
over $A$ is a triple $(X,\alpha,V)$ where $X$ is any set, $\alpha \in \fun X$
and $V : A \rightarrow \psf(X)$ is a valuation.
\end{defi}

The semantics of formulas in the modal one-step language in a one-step model is
given as follows:
\begin{itemize}
\item standard clauses for the boolean connectives,
\item $(X,\alpha,V) \vDash_1 \lambda(\psi_1,...,\psi_n)$ iff 
   $\alpha \in \lambda_X(\Vert \psi_1\Vert,...,\Vert \psi_n \Vert)$, where 
   $\Vert \psi_i \Vert \subseteq X$ is the (classical) truth set of the formula 
   $\psi_i$ under the valuation $V$.
\end{itemize}
We can now define the class of automata used to characterize the coalgebraic 
$\mu$-calculus.

\begin{defi}
Let $P$ be a finite set of variables and $\Lambda$ a set of predicate liftings. 
Then a \textit{($P$-chromatic) modal $\Lambda$-automaton} is a tuple 
$(A,\Delta,\Omega,a_I)$ where $A$ is a finite set of states with $a_I \in A$,
\[
\Delta : A \times \psf(P) \rightarrow \MLLa^{1}(A)
\]
is the transition map of the automaton, and $\Omega : A \rightarrow \omega$ is 
the parity map. 
The class of these automata is denoted as $\Aut(\MLLa)$.
\end{defi}

The acceptance game for an automaton $\mathbb{A} = (A,\Delta,\Omega,a_I)$ and 
a $\fun$-model $(S,\sigma,V)$ is given by the following table:

\begin{table}[h]
    \centering
\begin{tabular}{|l|c|l|}
\hline
Position  & Pl'r  &  Admissible moves 
\\ \hline
     $(a,s)\in A\times S$  
   & $\exists$  
   & $\{U:A\rightarrow {\mathcal{P}}S \mid 
        (S,\sigma(s),U) \vDash_{1}\Delta(a, V^{\dagger}(s))$
\\
     $U:A\rightarrow{\mathcal{P}}S$ 
   & $\forall$ 
   & $\{(b,t)\mid t\in U(b)\}$                                                    \\
\hline
    \end{tabular}
\end{table}

The loser of a finite match is the player who got stuck, and the 
winner of an infinite match is $\exists$ if the greatest parity that 
appears infinitely often in the match is even, and the winner is 
$\forall$ if this parity is odd.  
The automaton $\mathbb{A}$ \emph{accepts} the pointed model $(\smod,s)$ 
if $\exists$ has a winning strategy in the acceptance game from the starting
position $(a_I,s)$. 
We say that and automaton $\mathbb{A}$ is \emph{equivalent} to a formula 
$\varphi \in \muMLLa$ if, for every pointed $\fun$-model 
$(\mathbb{S},s)$, we have that $\mathbb{A}$ accepts $(\mathbb{S},s)$ iff 
$(\mathbb{S},s)\vDash \varphi$.


\begin{fact}\cite{font:auto10}
\label{coalgebraicmu}
Let $\fun$ be a set functor, and $\Lambda$ a set of monotone predicate 
liftings for $\fun$, closed under Boolean duals. 
Then
\[
\muMLLa \equiv \Aut(\MLLa).
\]
That is, there are effective transformations of formulas in $\muMLLa$ into
equivalent automata in $\Aut(\MLLa)$, and vice versa.
\end{fact}

\section{Coalgebraic $\MSO$}
\label{sec:mso}

We now introduce coalgebraic monadic second-order logic for a set functor $\fun$
and a set of liftings $\La$ and show how $\MSO$ can be recovered as a special 
case.
We define the syntax of the monadic second-order logic $\MSO_\fun$ by the 
following grammar:
\[
\varphi \isbnf
  \bot \divbnf  \sr(p) \divbnf p \subseteq q  
  \divbnf \lambda(p,q_1,..,q_n) \divbnf
  \varphi \vee \varphi \divbnf \neg \varphi \divbnf \exists p. \varphi
\]
where $\lambda$ is any $n$-place monotone predicate lifting and $p,q,q_1,...,q_n \in Var$. 
More generally, restricting to a set $\La$ of monotone liftings for $\fun$,
we define the sublanguage $\MSOLa \subseteq \MSOT$ by the same grammar 
except that we require the liftings to be in $\La$.

For the semantics, let $(\mathbb{S},s)$ be a pointed  $\fun$-model. 
We define the satisfaction relation ${\vDash} \subseteq S \times 
\MSO_\fun$ as follows:
\begin{itemize}
\item $(\mathbb{S},u) \vDash \sr(p)$ iff $V(p) = \{u\}$,
\item $(\mathbb{S},u)\vDash p \subseteq q$ iff  $V(p) \subseteq V(q)$,
\item $(\mathbb{S},u) \vDash \lambda(p, q_1,...,q_n)$ iff $\sigma (v)\in 
   \lambda_S(V(q_1),..,V(q_n))$ for all $ v \in V(p)$,
\item standard clauses for the Boolean connectives
\item $(\mathbb{S},u) \vDash \exists p .\varphi$ iff
   $(S,\sigma,V[p \mapsto Z],u)\vDash  \varphi$, some $Z \subseteq S$.
\end{itemize}

\noindent
We introduce the following abbreviations:
\begin{itemize}
\item $p = q$ for $p \subseteq q \wedge q \subseteq p$,
\item $\Em(p)$ for $\forall q. (q \subseteq p \rightarrow q = p)$,
\item $\Sing(p)$ for $\neg \Em (p) \wedge 
   \forall q (q \subseteq p \rightarrow (em(q) \vee q = p))$
\end{itemize}
expressing, respectively, that $p$ and $q$ are equal, that $p$ denotes the
empty set, and that $p$ denotes a singleton.

Clearly, standard $\MSO$ is the logic $\MSO_{\{\Diamond\}}$, where $\Diamond$
is the predicate lifting corresponding to the usual diamond modality over 
Kripke models. 
Obviously then, $\MSO_\psf$ contains $\MSO$. 
In order to see that the languages are in fact equivalent in expressive power,
we need the notion of \emph{expressive completeness}, which plays an important 
role in this paper.

\begin{defi}
A set of monotone liftings $\La$ for a set functor $\fun$ is said to be 
\textit{expressively complete} if, for every finite set of variables $A$ and 
every monotone predicate lifting $\lambda : \cvp(-)^A \to \cvp \circ \fun $, 
there exists a formula $\varphi \in \MLLa^1(A)$ such that, for every one-step
model $(X,\alpha,V)$ with $V : A \to \cvp (X)$, we have
$$
(X,\alpha,V)\vDash_1 \varphi \text{ iff } \alpha \in \lambda_X(V).
$$
\end{defi}

If $\La$ is expressively complete, then clearly $\muMLLa$ is equivalent in
expressive power to the full language $\muMLT$. 
It is not much harder to show that, under the same conditions, $\MSOLa$ is
equivalent in expressive power to the full language $\MSO_\fun$. 
Furthermore, expressive completeness can often be obtained fairly easily if
we make use of an application of the Yoneda lemma to represent $n$-place
predicate liftings as subsets of $\fun (2^n)$, a method developed 
in~\cite{schr:expr08}. 
In particular, since the liftings $ \{\Box,\Diamond\}$ for $\psf$ are 
expressively complete and $\Box$ is clearly definable in $\MSO_{\{\Diamond\}}$,
one can show that $\MSO = \MSO_{\{\Diamond\}}$ is equivalent in expressive 
power to the full coalgebraic logic $\MSO_{\psf}$. 
Furthermore, $\muML_{\psf}$ is equivalent to $\muML_{\{\Box,\Diamond\}}$. 
As a second example, involving the monotone neighborhood functor $\mon$, let 
$\Box$ here be the predicate lifting defined by $\alpha \in \Box_X(Z)$ iff 
$Z \in \alpha$, and let $\Diamond$ be its dual.
Then  the language $\MSO_\mon$ is equivalent to $\MSO_{\{\Box,\Diamond\}}$, and 
also $\muML_{\mon}$ is equivalent to $\mu \ML_{\{\Box,\Diamond\}}$.

Finally, as mentioned in the introduction, the key question in this paper will 
be to compare the expressive power of coalgebraic monadic second-order logic
to that of the coalgebraic $\mu$-calculus.
The following observation, of which the (routine) proof is omitted, provides
the easy part of the link.

\begin{prop}
\label{p:mu-to-mso}
Let $\La$ be a set of monotone predicate lfitings for the set functor $\fun$.
There is an inductively defined translation $(\cdot)^{\diamond}$ mapping any
formula $\phi \in \muMLLa$ to an equivalent formula $\phi^{\diamond}\in \MSOLa$.
\end{prop}

\section{Automata for coalgebraic $\MSO$}
\label{sec:aut}

In this section, we introduce automata for coalgebraic monadic second-order 
logic.

\subsection{A general perspective on parity automata}

Standard monadic second-order formulas can be translated to equivalent automata
over \textit{trees}, but this equivalence is not guaranteed to extend to
arbitrary Kripke models. 
In the case of general coalgebra, we should expect having to introduce a 
coalgebraic concept of ``tree-like'' models.

\begin{defi}
Given a set $S$ and $\alpha \in \fun S$, a subset $X \subseteq S$ is said to
be a \textit{support} for $\alpha$ if there is some $\beta \in \fun X$ with 
$\fun \iota_{X,S}(\beta) = \alpha$. A \textit{supporting Kripke frame} for a 
$\fun$-coalgebra $(S,\sigma)$ is a binary relation $R \subseteq S \times S$ 
such that, for all $u \in S$, $R(u) = \{v \mid u R v\}$ is a support for 
$\sigma(u)$. 
\end{defi}

\begin{defi}
A $\fun$-\textit{tree model} is a structure $(\mathbb{S},R,u)$ where 
$\mathbb{S} = (S,\sigma,V)$ is a $\fun$-model and $u \in S$, such that $R$ is 
a supporting Kripke frame for the coalgebra $(S,\sigma)$, and furthermore 
$(S,R)$ is a tree rooted at $u$, so that there is a unique $R$-path from $u$ 
to $w$ for each $w \in S$.
\end{defi}

Our goal is to translate formulas in $\MSOT$ to equivalent automata over 
$\fun$-tree models.
We start by introducing a very general type of automaton, originating 
with~\cite{vene:expr14}. 

\begin{defi}
Given a finite set $A$, a \textit{generalized predicate lifting} over $A$ 
comprises an assignment of a map
$$\varphi_X : (\cvp X)^A  \rightarrow \cvp \fun X.
$$
to every set $X$.
Concepts like \emph{Boolean dual} and \emph{monotonicity} apply to these 
liftings in the obvious way.
\end{defi}

The difference with respect to standard predicate liftings is that the 
components of a generalized predicate lifting do not need to form a natural 
transformation.\footnote{%
   In the style of abstract logic, it would make sense to require a
   general predicate lifting to be natural with respect to certain 
   maps, in particular, bijections. 
   For the purpose of this paper such a restriction is not needed, 
   however.
   }

\begin{defi}
A \textit{one-step language} $\mathlang{L}$ consists of a collection 
$\mathlang{L}(A)$ of generalized predicate liftings for every finite set $A$.  
The semantics of a generalized predicate lifting $\varphi$ in a one-step model 
$(X,\alpha,V)$ is given by 
$$
(X,\alpha,V)\vDash_1 \varphi \textit{ iff } \alpha \in \varphi_X(V).
$$
\end{defi}

Our automata will be indexed by a (finite) set of variables involved, 
corresponding to the set of free variables of the $\MSOT$-formula.

\begin{defi}
Let $P \subseteq \Var$ be a finite set of variables and let $\mathlang{L}$ be a 
one-step language for functor $\fun$. 
A \textit{($P$-chromatic) $\mathlang{L}$-automaton} is a structure 
$(A,\Delta,\Omega,a_I)$ where
\begin{itemize}
\item $A$ is a finite set, with $a_I \in A$,
\item $\Omega : A \rightarrow \omega$ is a parity map, and
\item $\Delta : A \times \psf(P) \rightarrow \mathlang{L}(A) $ is the 
  transition map of $\bbA$.
\end{itemize}
\end{defi}

The \textit{acceptance game} of $\mathbb{A}$ with respect to a $\fun$-tree model 
$(T,R,\sigma,V,u)$ is given by Table~\ref{table:accgame}.
We say that the automaton $\mathbb{A}$ accepts the model $(T,R,\sigma,V,u)$ 
if $\exists$ has a winning strategy in this game (initialized at position
$(a_{I},u)$).

\begin{table*}[ht]
{\normalsize
\centering
\begin{tabular}{|l|c|l|c|}
\hline
Position & Player & Admissible moves & Parity \\
\hline
    $(a,s) \in A \times T$
  & $\exists$
  & $\{U : A \to \psf(R(s)) \mid  $ &  \\
& &  $(R(s),\sigma(s), U) 
               \vDash_{1} \Delta(a,V^{\dag}(s))  \}$
  & $\Omega(a)$
\\
    $U : A \rightarrow \psf(T)$
  & $\forall$
  & $\{(b,t) \mid t \in U(b) \}$
  & $0$
\\ \hline
   \end{tabular}
   \caption{\small Acceptance game for parity automata.}
\label{table:accgame}
}
\end{table*}

\subsection{Closure properties}

This abstract level is useful for establishing some simple closure properties 
of automata, based on properties of the one-step language. 
The first, easy, results establish sufficient conditions for closure under 
union and complementation.

\begin{prop}
\label{closureunion}
If the one-step language $\mathlang{L}$ is closed under disjunction,
then the class of $\mathlang{L}$-automata is closed under union.
\end{prop}

\begin{prop}
\label{closurecomplementation}
If the monotone fragment of the one-step language $\mathlang{L}$ is closed under
Boolean duals, then the class of $\mathlang{L}$-automata is closed under 
complementation.
\end{prop}

The most interesting property concerns closure under existential projection. 
The following terminology is taken from~\cite{jani:expr96}, but instead
of relying on a particular syntactic shape of one-step formulas, we define
the concepts in purely semantic terms.

\begin{defi}
A predicate lifting $\varphi$ over $A$ is said to be \textit{special basic} if,
for every one-step model $(X,\alpha,V)$ such that 
$$(X,\alpha,V)\vDash_1 \varphi$$
there is a valuation $V^* : A \rightarrow \cvp(X)$ such that
\begin{itemize}
\item $V^\ast(a) \subseteq V(a)$ for each $a \in A$,
\item $V^\ast(a) \cap V^\ast(b) = \emptyset$ whenever $a \neq b$, and
\item $(X,\alpha,V^*) \vDash_1 \varphi$. 
\end{itemize}
Call an $\mathlang{L}$-automaton \textit{non-deterministic} if every lifting 
$\Delta(a,c)$ is special basic.
\end{defi}

It is easy to see that if the language $\mathlang{L}$ is closed under 
disjunctions, then so is its fragment of special basic liftings.
From this we obtain the following.

\begin{prop}
\label{existentialclosure}
If the one-step language $\mathlang{L}$ is closed under disjunction, then the
class of non-deterministic $\mathlang{L}$-automata is closed under existential
projection over $\fun$-tree models. 
\end{prop}

\begin{proof}
Suppose $\mathbb{A} = (A,\Delta,a_I,\Omega)$ is a non-deterministic
$\mathlang{L}$-automaton for the variable set $P$. 
Define the $P\setminus q$-chromatic automaton 
$\exists q.\mathbb{\mathbb{A}} = (A,\Delta^*,a_I,\Omega)$ by setting
$$
\Delta^*(a,c) = \Delta(a,c) \vee \Delta(a,c\cup\{q\}).
$$
It is easy to see that every $\fun$-tree model accepted by $\mathbb{A}$ is 
also accepted by $\exists p. \mathbb{A}$.
Conversely, suppose $\exists p.\mathbb{A}$ accepts some $\fun$-tree 
model $(S,R,\sigma,V,s_I)$. For each winning position $(a,s)$ in the 
acceptance game, let $V_{(a,s)}$ be the valuation chosen by $\exists$ 
according to some given winning strategy $\chi$. Note that we can assume that $\chi$ is a \textit{positional} winning strategy, since $\exists p. \mathbb{A}$ is a parity automaton.
It is not difficult to see that the 
automaton $\exists p .\mathbb{A}$ is a non-deterministic automaton, 
and so for each winning position $(a,s)$ there is a valuation $V_{(a,s)}^* : 
A \rightarrow \psf(R(s))$, which is an admissible move for $\exists$, such that 
$V_{(a,s)}^*(b) \subseteq V_{(a,s)}(b)$ and such that for all $b_1 \neq b_2 
\in A$ we have $V_{(a,s)}^*(b_1) \cap V_{(a,s)}^*(b_2) = \emptyset$.
Define the strategy $\chi^*$ by letting $\exists$ choose the valuation 
$V_{(a,s)}^*$ at each winning position $(a,s)$  - this is still a winning 
strategy, since the valuations chosen by $\exists$ are smaller and so no new 
choices for $\forall$ are introduced. Furthermore, $\chi^*$  is clearly still a positional winning strategy.

From these facts follow by a simple induction on the height of the nodes in the supporting
tree that the strategy $\chi^*$ is \textit{scattered}, i.e. that for every $s \in S$
there is at most one automaton state $a$ such that $(a,s)$ appears in a 
$\chi^*$-guided match of the acceptance game. 
So we can define a valuation $V^\prime$ like $V$ except we evaluate $q$ to be
true at all and only the states $s$ such that 
$$
(R(s),\sigma(s), V^*_{(a_s,s)})\vDash_1 \Delta(a_s,c \cup \{q\}),
$$
where $a_s$ is a necessarily \emph{unique} automaton state such that $(a,s)$
appears in some $\chi^*$-guided match, and $c$ is the color consisting of the
variables true under $V$ at $s$.
It is not hard to show that $\mathbb{A}$ accepts $(S,R,\sigma,V',s_I)$.
\end{proof}

\subsection{Second-order automata}

We now introduce a more concrete one-step language for a given set functor 
$\fun$ and a given set of (natural) liftings $\Lambda$, and show that 
$\MSOLa$ can be translated into the corresponding class of automata. 

Let $\Lambda$ be a set of monotone predicate liftings for $\fun$.
The set  of \textit{second-order one-step formulas} over any
set of variables $A$ and relative to the set of liftings $\Lambda$
is defined by the grammar:
$$\varphi \isbnf a \subseteq b \divbnf \lambda(a_1,...,a_n) \divbnf \neg \varphi 
   \divbnf \varphi \lor \varphi 
   \divbnf \exists a.\varphi,
$$
where $a,b,a_1,...,a_n \in A$ and $\lambda$ is any predicate lifting in 
$\Lambda$. 
Fixing an infinite set of ``one-step variables'' $Var_1$, and given a finite set
$A$, the set of \textit{second-order one-step sentences} over $A$,
denoted $\SOLa^{1}(A)$, is the set of one-step formulas over $A \cup Var_1$, 
with all free variables belonging to $A$.
We write $\SOT^{1}(A)$ when $\La$ comprises all monotone liftings for $\fun$.

The semantics of a one-step second-order $A$-formula in a one-step model 
$(X,\alpha,V)$ (with $V : A \to \psf(X)$ is defined by the following clauses:
\begin{itemize}
\item $(X,\alpha,V)\vDash_1 p \subseteq q$ iff $V(p) \subseteq V(q)$,
\item $(X,\alpha,V)\vDash_1 \lambda (p_1,...,p_n)$ iff $\alpha \in 
    \lambda_X(V(p_1),...,V(p_n))$,
\item standard clauses for the Boolean connectives,
\item $(X,\alpha,V)\vDash_1 \exists p. \varphi$ iff 
    $(X,\alpha,V[p\mapsto S])\vDash_1 \varphi$ for some $S \subseteq X$.
\end{itemize}

Any one-step second-order $A$-sentence $\phi$ can be regarded as a generalized 
predicate lifting over $A$, with
$$
\varphi_X(V) = \{\alpha \in \fun X \mid (X,\alpha,V) \vDash_1 \varphi\}.
$$
Note that the syntax of $\SOT^{1}$ allows negations, implying that not all 
these predicate liftings are monotone.

\begin{defi}
Let $\Lambda$ be a set of monotone predicate liftings for $\fun$.
A \textit{second-order $\Lambda$-automaton} is an $\mathlang{L}$-automaton for
$\mathlang{L}$ being the assignment of the one-step second-order $A$-sentences 
$\SOLa^{1}(A)$ to every set of variables $A$. 
We write $\Aut(\SOLa)$ to denote this class, and $\Aut(\SOT)$ in case $\La$ is
the set of \emph{all} monotone predicate liftings for $\fun$.
\end{defi}

Our aim is to prove that every formula of $\MSOLa$ can be
translated into an equivalent second-order $\Lambda$-automaton (over rooted 
$\fun$-tree models), and the main problem here is to obtain closure under 
existential projection.

The key to this step is a simulation theorem. 
First, a useful trick due to Walukiewicz \cite{walu:mona96} allows us to 
transform any second-order automaton into one in which all the one-step formulas are monotone, when regarded as generalized predicate liftings. We call such an automaton a \textit{monotone} automaton.

\begin{prop}
\label{p:mon-aut}
Let $\La$ be any set of monotone predicate liftings.
Every automaton $\bbA \in \Aut(\SOLa)$ is equivalent to a monotone second-order
$\bbA \in \Aut(\SOLa)$.
\end{prop}
\begin{proof}
Enumerate $A$ as $\{a_1,...,a_k\}$, and just replace each formula $\Delta(a,c)$ 
by
$$
\exists Z_1 ... \exists Z_k. Z_1 \subseteq a_1 \wedge ... \wedge 
  Z_k \subseteq a_k \wedge \Delta(a,c)[Z_i / a_i]
$$
where $\Delta(a,c)[Z_i/a_i]$ is the result of substituting the variable
$Z_i$ for each open variable $a_i$ in $\Delta(a,c)$. 
This new formula is monotone in the variables $A$ and the resulting automaton
is equivalent to $\mathbb{A}$.
\end{proof}

The intuition behind the simulation theorem is the same as that behind the 
standard ``powerset construction'' for word automata: the states of the new 
non-deterministic automaton $\mathbb{A}_n$ are ``macro-states'' representing
several possible states of $\mathbb{A}$ at once. 
Formally, the states of $\mathbb{A}_n$ will be binary relations over $A$, and 
given a macro-state $R$, its range gives an exact description of the states in
$\mathbb{A}$ that are currently being visited simultaneously. It is safe to 
think of the macro-states as subsets of $A$, however: the only reason that we
have binary relations over $A$ as states rather than just subsets is to have 
a memory device so that we can keep track of traces in infinite matches. 
For each macro-state $R$ and each colour $c$ we want to be able to say that the 
one-step formulas corresponding to each state in the range of $R$ hold, so we
want to translate the one-step formulas over $A$ into one-step formulas over 
the set of macro-states. In order to translate a formula $\Delta(a,c)$ to a 
new one-step formula with macro-states as variables, we have to replace the 
variable $b$ in $\Delta(a,c)$ with a new variable that acts as a stand-in for
$b$. For this purpose we introduce a new, existentally quantified variable 
$Z_b$, together with a formula stating explicitly that $Z_b$ is to represent
the union of the values of all those macro states that contain $b$. 
Furthermore we want all the one-step formulas to be special basic, and for
this purpose we simply add a conjunct ``$\mathsf{disj}$'' to each one-step
formula, stating that the values of any pair of distinct variables appearing
in the formula are to be disjoint.  Finally, in order to turn $\mathbb{A}_n$
into a parity automaton, we use a stream automaton to detect bad traces (see
for instance \cite{vene:lect12} for the details in a more specific case).

\begin{theo}[Simulation]
\label{simulationtheorem}
Let $\La$ be a set of monotone predicate liftings for $\fun$.
For any monotone automaton $\mathbb{A} \in \Aut({\SOLa})$ there exists an equivalent
non-deterministic $\mathbb{A}'\in \Aut({\SOLa})$.
\end{theo}

Given a set $A$, we consider the set $\psf(A \times A)$ as a set of variables.
Let
$$\mathsf{disj} : = \bigwedge_{B \neq B' \subseteq A\times A} \forall X . (X \subseteq B \wedge X \subseteq B') \rightarrow \mathtt{Em}(X)$$
Pick a fresh variable $Z_a$ for each $a \in A$. Given a 1-step formula $\varphi$, let
$$\varphi[Z_a/ a]$$ be the result of substituting $Z_a$ for each free variable $a \in A$ in $\varphi$.
If we enumerate the elements of $A$ as $a_1,...,a_k$, we now define the formula $\varphi^{\uparrow b}$ for $b \in A$ to be
\begin{displaymath}
\begin{array}{lcl}
 & \exists Z_{a_1}...\exists Z_{a_k}. \bigwedge_{1 \leq i \leq k}  ( \mathsf{Eq}(Z_{a_i}  \bigcup\{B^\prime \mid (b,a_i) \in B^\prime \})\; \wedge \\
& \varphi[Z_a/ a] )
\end{array}
\end{displaymath}
where  $\mathsf{Eq}(Z_{a_i}  \bigcup\{B^\prime \mid (b,a_i) \in B^\prime \})$ is a formula asserting that the value of the variable $Z_{a_i}$ is the union of the values of all variables $B^\prime$ with $(b,a_i)\in B^\prime$.

Let $\mathbb{A} = (A,\Delta,a_I,\Omega)$ be any monadic $\Lambda$-automaton. We can assume w.l.o.g. that $\mathbb{A}$ is monotone. We first construct the automaton $\mathbb{A}_n = (A_n,\Delta_n,a_I^*,F)$ with a non-parity acceptance condition $F \subseteq (A_n)^\omega$ as follows:
\begin{itemize}
\item $A_n = \psf(A \times A)$
\item $\Delta_n(B,c) = \mathsf{disj} \wedge  \bigwedge_{b \in \pi_2[B]}\Delta(b,c)^{\uparrow b}$
\item $a_I^\ast = \{(a_I,a_I)\}$
\item $F$ is the set of streams over $\psf(A \times A)$ with no bad traces.
\end{itemize}
Here, $\pi_2$ is the second projection of a relation $B$ so that $\pi_2[B]$ denotes the range of $B$. A \textit{trace} in a stream $(B_1,B_2,B_3,...)$ over $\psf(A \times A)$ is a stream $(a_1,a_2,a_3,...)$ over $A$ with $a_1 \in \pi_2[B_1]$ and $(a_{j-1},a_j) \in B_j$ for $j > 1$. A trace is \textit{bad} if the greatest number $n$ with $\Omega(a_i) = n$ for infinitely many $a_i$ is odd.

The new automaton $\mathbb{A}_n$ is clearly special basic, because of the $\mathsf{disj}$-formulas.

\begin{lemma}
$\mathbb{A}_n$ is equivalent to $\mathbb{A}$, provided that $\mathbb{A}$ is monotone
\end{lemma}

\begin{proof}
Fix a pointed $\fun$-tree model $(\mathbb{S},R,s_I)$ where $\mathbb{S} = (S,\sigma,V)$. We want to show that $\mathbb{A}$ accepts $(\mathbb{S},R,s_I)$ if and only if $\mathbb{A}_n$ does. That is, we want to show that the languages $L(\mathbb{A})$ and $L(\mathbb{A}_n)$ defined by these two automata are the same.

\subsubsection*{First part: $L(\mathbb{A}) \subseteq L(\mathbb{A}_n)$}

Suppose first that $\mathbb{A}$ accepts $(\mathbb{S},R,s_I)$. Let $\chi$ be a positional winning strategy for $\exists$ in the acceptance game, mapping each winning position $(a,s)$ to a valuation $U : A \rightarrow \cvp(R(s))$ such that
$$(R(s),\sigma(s),U) \vDash_1 \Delta(a,V^\dagger(s))$$
Such a strategy exists since $\mathbb{A}$ is a parity automaton, and so the acceptance game is a parity game.
We define the winning strategy $\chi^*$ for $\exists$ in the acceptance game for $\mathbb{A}_n$ as follows: given a position $(B,s)$, define the function $f_{B,s} : R(s) \rightarrow \psf(A \times A)$ by setting
$$f_{B,s} (s^\prime) = \{(a,b) \mid a \in \pi_2[B] \;~\&~\;s^\prime \in \chi(a,s)(b)\}$$
At the position $(B,s)$, let $\exists$ choose the following valuation $\chi^*(B,s)$, defined by:
$$\chi^*(B,s)(B^\prime) = \{s^\prime \in R(s) \mid f_{B,s}(s^\prime) = B^\prime\}$$
Our first claim is that, for each position of the form $(B,s)$ where each $(b,s)$ for $b \in \pi_2[B]$ appears in some $\chi$-coherent match of the acceptance game for $\mathbb{A}$ with start position $(a_I,s_I)$, the move for $\exists$ given by the strategy $\chi^*$ is legal. To prove this claim we need to check that, for each position $(B,s)$, we have
$$(R(s),\sigma(s),\chi^*(B,s))\vDash_1 \Delta_n(B,V^\dagger(s))$$
provided each $(b,s)$ for $b\in \pi_2[B]$ appears in some $\chi$-coherent match. First, the formula $\mathsf{disj}$ is true since the marking $\chi^*(B,s)$ is the inverse of a mapping from $R(s)$ to $\psf(A \times A)$. We now have to check that,
for each $a^\prime \in \pi_2[B]$ we have
  $$(R(s),\sigma(s),\chi^*(B,s)) \vDash_1 \Delta(a^\prime,V^\dagger(s))^{\uparrow a^\prime} $$
	We need to find sets $S_{a_1},...,S_{a_k} \subseteq R(s)$ such that $(R(s),\sigma(s),\chi^*(B,s))$ with the assignment $Z_{a_i} \mapsto S_{a_i}$ satisfies the formula
 \begin{displaymath}
\begin{array}{lcl}
 \bigwedge_{1 \leq i \leq k}  \mathsf{Eq} (Z_{a_i}, \bigcup\{B' \mid (a^\prime,a_i) \in B' \} )\; &  \wedge & \\
 \Delta(a^\prime,V^\dagger(s))[Z_a/a] & &
\end{array}
\end{displaymath}
Since $\chi$ gives a legal move at the position $(a^\prime,s)$ for each $a' \in \pi_2[B]$, each one-step model of the form $(R(s),\sigma(s),\chi(a^\prime,s))$ satisfies the formula $\Delta(a^\prime,V^\dagger(s))$.
 Hence, if we assign to each variable $Z_{a_i}$ the set $\chi(a^\prime,s)(a_i)$, then this variable assignment satisfies the formula
$$\Delta(a^\prime,V^\dagger(s))[a \mapsto Z_a \mid a \in A]$$
Since the formula $\Delta(a^\prime,s)$ is monotone in all the variables $A$, the same is true for any larger assignment. So
it now suffices to prove that
$$\chi(a^\prime,s)(a_i) \subseteq \bigcup \{\chi^*(B,s)(B^\prime) \mid (a^\prime,a_i) \in B^\prime\}$$
since we can then safely take
$$S_{a_i} = \bigcup \{\chi^*(B,s)(B^\prime) \mid (a^\prime,a_i) \in B^\prime\}$$
To prove this inclusion, suppose
$s^\prime \in \chi(a^\prime,s)(a_i)$. Let $B^\prime$ be the relation defined by
$$ (d,d^\prime) \in B^\prime \Leftrightarrow d \in \pi_2[B] \;\& \; s^\prime \in \chi(d,s)(d^\prime)$$
Clearly, $(a^\prime,a_i) \in B^\prime$. Moreover,  $f_{B,s}(s^\prime) = B^\prime$ by definition, and so $s^\prime \in \chi^*(B,s)(B^\prime)$ as required.

We now show that any $\chi^*$-coherent match with start position $(a_I^*,s_I)$ is winning for $\exists$. We have to prove two things: first, that $\exists$ never gets stuck in a $\chi^*$-coherent match, and second, that $\exists$ wins every infinite $\chi^*$-coherent match, i.e. no infinite $\chi^*$-coherent match contains a bad trace.

First we show that $\exists$ never gets stuck. For this to be the case, all we need to show that if $(B,s)$ is the last position of some $\chi^*$-coherent partial match, then all the positions $(a,s)$ for $a \in \pi_2[B]$ are winning positions for the strategy $\chi$ - by our previous claim this guarantees the move $\chi^*(B,s)$ to be legal. We prove by induction on the length of a finite partial match that this holds for the last position of the match: it holds for $(\{(a_I,a_I)\},s_I)$, clearly, since $\chi$ is a winning strategy at $(a_I,s_I)$. Suppose that the induction hypothesis holds for a finite match with last position $(B,s)$. Let $(B',s')$ be any position such that $s' \in \chi^*(B,s)(B')$. Then
$$B' = f_{B,s}(s^\prime) = \{(a,b) \mid a \in \pi_2[B] \;~\text{and}~ \;s^\prime \in \chi(a,s)(b)\}$$
So suppose $b\in \pi_2[B']$. Then there is some $a$ with $(a,b) \in B^\prime$, and we must have $a \in \pi_2[B]$ and $s^\prime \in \chi(a,s)(b)$. But since the position $(a,s)$ is winning by the inductive hypothesis, this means that $(b,s^\prime)$ is a winning position for $\chi$, and we are done.

We now show that $\exists$ wins every infinite $\chi^*$-coherent match. For this, it suffices to show that every trace $(a_1,a_2,a_3,...)$ in a $\chi^*$-coherent infinite match
$$(B_1,s_1),(B_2,s_2),(B_3,s_3),...$$
corresponds to a $\chi$-coherent match
$$(a_I,s_I) = (a_1,s_1),(a_2,s_2),(a_3,s_3),...$$
So fix a trace $(a_1,a_2,s_3,...)$ meaning that for each $a_i$ we have $(a_i,a_{i + 1}) \in B_{i + 1}$. We have to show that $s_{i + 1} \in \chi(a_i,s_i)(a_{i + 1})$ for each $i$. We have
$$s^{i + 1} \in \chi^*(B_{i},s_i)(B_{i + 1})$$ meaning that $B_{i + 1}$ is equal to
$$ f_{B_i,s_i}(s_{i + 1}) = \{(a,b) \mid a \in \pi_2[B_i] \;\& \;s_{i + 1} \in \chi(a,s_i)(b)\}$$
In particular, since $(a_i,a_{i + 1}) \in B_{i + 1}$, this means we must have $s_{i +1} \in \chi(a_i,s_i)(a_{i + 1})$ as required.

\subsubsection*{Second part: $L(\mathbb{A}_n) \subseteq L(\mathbb{A})$}

Conversely, suppose $\mathbb{A}_n$ accepts $(\mathbb{S},s_I)$ with winning strategy $\chi$. We construct a winning strategy $\chi^*$ for $\exists$ w.r.t $\mathbb{A}$. Note that the strategy $\chi$ is not necessarily positional, since the acceptance game is not a parity game.

By induction on the length of a $\chi^\ast$-coherent partial match
$$M = (a_1,s_1),(a_2,s_2),(a_3,s_3)...(a_{k},s_{k})$$
with $(a_1,s_1) = (a_I,s_I)$,
we are going to define a next legal move $\chi^*(M)$ for $\exists$, and by a simultaneous induction we construct a $\chi$-coherent partial match
$$N = (B_1,s_1),(B_2,s_2),(B_3,s_3)...(B_{k},s_{k})$$
with $(B_1,s_1) = (a_I^*,s_I)$, $a_j \in \pi_2[B_j]$ for each $j$ and $(a_{j- 1},a_{j}) \in B_j$ for each $k \geq j > 1$.
Furthermore we will make sure that whenever a $\chi^*$-coherent match $M$ is an initial segment of a match $M'$, the $\chi$-coherent match associated with $M$ is an initial segment of the $\chi$-coherent match associated with $M'$.  It will follow at once that $\chi^*$ is a winning strategy, since $\exists$ never gets stuck in any $\chi^*$-coherent partial match and, furthermore, every infinite $\chi^*$-coherent match corresponds to a trace in some $\chi$-coherent infinite match.

The base case of the induction is the unique match of length 1 with the single position $(a_I,s_I)$, and we take the corresponding position to be $(a_I^* ,s_I)$. Now, suppose $\chi^* $ has been defined on all matches of length $< k$, and let $M$ be a $\chi^*$-coherent partial match of length $k$ of the form
$$(a_1,s_1),(a_2,s_2),(a_3,s_3)...(a_{k},s_{k})$$
By the inductive hypothesis we have a corresponding $\chi$-coherent match $N$ which we write as
$$(B_1,s_1),(B_2,s_2),(B_3,s_3)...(B_{k},s_{k})$$
with $a_k \in \pi_2[B_k]$. Now we define the next legal move $\chi^*(M)$ for $\exists$, and we show that for every position $(a',s')$ such that $s'\in \chi^*(M)(a')$, we can find a relation $B'$ such that $(a_k,a') \in B'$ and
$$ (B_1,s_1),(B_2,s_2),(B_3,s_3)...(B_{k},s_{k}),(B',s')$$
is a $\chi$-coherent match.

Since  $N$ is a $\chi$-coherent partial match and $\chi$ is a winning strategy for $\exists$, we have that
$$(R(s_k),\sigma(s_k),\chi(N)) \vDash_1 \Delta_n(B_k,V^\dagger(s_k))$$
Since $a_k \in \pi_2[B_k]$,
this means that there exist sets $S_b \subseteq R(s_k)$ for each $b \in A$ such that the 1-step model $(R(s_k),\sigma(s_k),\chi(N))$ satisfies the formula
\begin{displaymath}
\begin{array}{lcl}
\bigwedge_{b \in A}  \mathsf{Eq}(Z_{b}, \bigcup\{B' \mid (a_k,b) \in B' \} )\; & \wedge &  \\
 \Delta(a_k,V^\dagger(s_k))[Z_b / b] & &
\end{array}
\end{displaymath}
under the assignment $Z_b \mapsto S_b$. Hence, the valuation $U$ defined by
$$U : b \mapsto S_b$$
will be such that
$$(R(s_k),\sigma(s_k),U) \vDash_1 \Delta(a_k,V^\dagger(s_k))$$
So we set $\chi^*(M) = U$, a legal move. Note that we have
$$U(b) = \bigcup\{\chi(N)(B') \mid (a_k,b) \in B'\}$$
Now, let $(a',s')$ be such that $s' \in U(a')$. This means that there is some $B'$ with $(a_k,a')\in B'$ and $s' \in \chi(N)(B')$. Hence, $(B',s')$ satisfies the required conditions, and we are done.
\end{proof}
The only thing left to do at this point is to transform this automaton into one that has its acceptance condition given by a parity map. The set of streams over $\psf(A \times A)$ that contain no bad traces w.r.t. the parity map $\Omega$ is an $\omega$-regular stream language, so let
$$\mathbb{Z} = (Z,\delta,z_I,\Omega_z)$$
be a parity stream automaton, $\delta : Z \times \psf(A \times A) \rightarrow Z$, that recognizes this language. We now construct the automaton
$$\mathbb{A}_n \odot \mathbb{Z} = (A_n',\Delta_n',a_I',\Omega_n')$$
as follows:
\begin{itemize}
\item $A_n' = A_n \times Z$
\item $a_I' = (a_I^*,z_I)$
\item $\Omega_n ' (B,z) = \Omega(z)$
\item $\Delta_n'((B,z),c) = \Delta_n(B,s)[ (B',\delta(B,z))/ B']$
\end{itemize}
It is not difficult to check that
$\mathbb{A}_n \odot \mathbb{Z}$ is equivalent to $\mathbb{A}_n$.
Since $\mathbb{A}_n\odot \mathbb{Z}$ is clearly still a non-deterministic automaton, this ends the proof of the simulation theorem.

Combining Proposition \ref{existentialclosure} with 
Theorem~\ref{simulationtheorem},
we easily obtain the following closure property.

\begin{prop}
\label{p:clos-exist}
Let $\La$ be a set of monotone predicate liftings for a set functor $\fun$.
Over $\fun$-tree models, the class of second-order $\Lambda$-automata is closed
under existential projection.
\end{prop}

We can now use the closure properties we have established for second-order
automata to give the desired translation of $\MSOT$ into second-order automata.

\begin{prop}
\label{p:automatachar}
For every formula $\varphi\in \MSOLa$ with free variables in $P$, there exists 
a $P$-chromatic automaton $\mathbb{A}_\varphi \in \Aut(\SOLa)$ which is 
equivalent to $\varphi$ over $\fun$-tree models.
\end{prop}

\begin{proof}
Proceeding by a straightforward induction on the complexity of $\phi$, we leave
it to the reader to construct appropriate automata for the atomic formulas.
The inductive cases for disjunction and negation follow by the 
Propositions~\ref{closureunion} and~\ref{closurecomplementation}, together with
the easy observation that the one-step language $\SO_{\La}$ is closed under
disjunction and Boolean duals.
The case of existential quantification is taken care of by 
Proposition~\ref{p:clos-exist}.
\end{proof}

Theorem~\ref{t:automatachar} is immediate from this, as is the following.

\begin{coro}
\label{c:automatachar}
Suppose $\Lambda$ is any set of monotone predicate liftings for $\fun$ such that $\MSOT \equiv \MSO_\Lambda$.
Then for every formula of $\MSOT$, there exists an equivalent 
second-order $\Lambda$-automaton over $\fun$-tree models. In particular, this holds whenever $\Lambda$ is expressively complete.
\end{coro}

\section{Bisimulation invariance}
\label{sec:bisinv}

This section continues the program of~\cite{vene:expr14}, making use of the
automata-theoretic translation of $\MSOT$ we have just established. 
The gist of our approach is that, in order to characterize a coalgebraic 
fixpoint logic $\muMLT$ as the bisimulation-invariant fragment of $\MSOT$, it 
suffices to establish a certain type of translation between the corresponding 
one-step languages. 
First we need some definitions.

 %

\begin{defi}
Given sets $X,Y$, a mapping $h : X \rightarrow Y$ and a valuation $V : A 
\rightarrow \cvp (Y)$, we define the valuation $V_{[h]} : A \rightarrow \cvp (X)$ by 
setting $V_{[h]}(b) = h^{-1}[V(b)]$ for each $b \in A$.
\end{defi}

The most important concept that we take from $\cite{vene:expr14}$ is that
of a \textit{uniform translation} (called \textit{uniform correspondence} 
in \cite{vene:expr14}):

\begin{defi}
Given a functor $\fun$, a \textit{uniform construction} $F$ for 
$\fun$ takes each pair $(X,\alpha)$ with $\alpha \in \fun X$ to a tuple consisting
of a set $X_*$, an $\alpha_* \in \fun (X_*) $, 
and a map $h_\alpha :  X_* \rightarrow X $ such that 
$$\fun (h_\alpha)(\alpha_*) = \alpha$$
We say that the second-order one-step language $\SOLa^{1}(A)$ 
\textit{admits uniform translations} if, given any natural number $k$, there 
exists a uniform construction $F$ and an assignment of a monotone (natural)
predicate lifting 
$$\varphi^* : \cvp(-)^A \rightarrow \cvp \circ \fun$$
to each monotone one-step formula $\varphi \in \SOLa^{1}$ 
with free variables $A$ and quantifier depth at most $k$, such that for any 
one-step model $(X,\alpha,V)$, we have 
$$(X,\alpha,V) \vDash_1 \varphi^* \text{ iff } 
  (X_*,\alpha_*,V_{[h_\alpha]})\vDash_1 \varphi.
$$
\end{defi}

\begin{remark}
It is easy to see that every monotone predicate lifting $\lambda : \cvp(-)^A \to \cvp \circ \fun$ is equivalent to an atomic formula of $\MLT^{1}(A)$.
In the following we shall not take care to distinguish between such a monotone 
predicate lifting and the corresponding 
atomic formula.
\end{remark}

\begin{defi}
Any translation $(\cdot)^{*}: \SOLa^{1} \to \MLT^{1}$ induces a construction
on automata, transforming a second-order $\Lambda$-automaton
$\bbA = (A,\Delta,a_I,\Omega)$ into the modal automaton
$\bbA^{*} = (A,\Delta^{*},a_I,\Omega)$, with $\Delta^{*}$ given by 
$\Delta^{*}(a,c) \isdef (\Delta(a,c))^{*}$.
\end{defi}

The proof of the following result closely follows that of the main result 
in \cite{vene:expr14}.
The main difference with \cite{vene:expr14} is that here we need an
``unravelling''-like component.

\begin{prop}
\label{p:unr}
Assume that $\SOLa^{1}$ admits a uniform translation $(\cdot)^{*}$, and let 
$\bbA$ be a second-order $\Lambda$-automaton.
Then for each pointed $\fun$-model $(\bbS,s)$ there is a $\fun$-tree model 
$(\bbT,R,t)$, with a $\fun$-model homomorphism $f$ from $\bbT$ to
$\bbS$, mapping $t$ to $s$, and such that
\[
\bbA \text{ accepts } (\bbT,R,t) \text{ iff }
\bbA^{*} \text{ accepts } (\bbS,s).
\]
Furthermore, given that $\bbS = (S,\sigma,V)$, if the map $h_{\sigma(s)} : S_* \to S$ is surjective, so is $f$.
\end{prop}

\begin{proof}
Consider any given pointed $\fun$-model $(\mathbb{S}_1,s_1)$ where $\mathbb{S}_1 = (S_1,\sigma_1,V_1)$. We are going to construct a $\fun$-tree model $(\mathbb{S}_2,R,s_2)$, $\mathbb{S}_2 = (S_2,\sigma_2,V_2)$, together with a model homomorphism from the underlying pointed $\fun$-model $\mathbb{S}_2$  to $\mathbb{S}_1$ mapping $s_2$ to $s_1$, and such that $\mathbb{A}$ accepts the $\fun$-tree model $(\mathbb{S}_2,s_2)$ if and only if $\mathbb{A}^*$ accepts the pointed $\fun$-model $(\mathbb{S}_1,s_1)$. 

We construct this $\fun$-tree model as follows: for each $u \in S_1$, we define an associated pair $(X_u,\alpha_u)$ as follows: set $X_u = (S_1)_*$ and set $\alpha_u = \sigma_1(u)_*$.
Observe that, by the construction of these one-step models, for each $u \in S_1$, there is a mapping
$$\xi_u : X_u \rightarrow S_1$$
such that
\begin{enumerate}
\item $\fun(\xi_u)(\alpha_u) =  \sigma_1(u)$
\item For each valuation $U : A \rightarrow \cvp (S_1)$, every $u \in S_1$ and every one-step formula $\Delta(a,c)$ appearing in $\mathbb{A}$, we have
$$(S_1,\sigma_1(u),U)\vDash_1 \Delta(a,c) \text{ iff } (X_u,\alpha_u,U_{[\xi_u]})\vDash_1 \Delta^*(a,c)$$
\end{enumerate}
The map $\xi_u$ is given by $h_{\sigma_1(u)}$. We now construct the $\fun$-tree model $(S_2,R,\sigma_2,s_2,V_2)$ as follows: first, consider the set of all non-empty finite (non-empty) tuples $(v_1,...,v_n)$ of elements in 
$$\{s_1\}\cup\bigcup_{u \in S_1} X_u$$
such that $v_1 = s_1$. We define, by induction, for each natural number $n > 0$ a subset $M_n$ of this set, and a mapping $\gamma_n : M_n \rightarrow S_1$, as follows:
\begin{itemize}
\item set $M_1 = \{(s_1)\}$, and define $\gamma_1(s_1) = s_1$.
\item Set $M_{n + 1} = \{\vec{v} * w \mid \vec{v} \in M_n, w \in X_{\gamma_n(\vec{v})}\}$. Define $\gamma_{n + 1} (\vec{v} * w) = \xi_{\gamma_n(\vec{v})}(w)$.
\end{itemize}
Here, we write $\vec{v}*w$ to denote the tuple $(v_1,...,v_n,w)$ if $\vec{v} = (v_1,...,v_n)$. 
Set $S_2 = \bigcup_{n > 0} M_n$, and define $\gamma = \bigcup_{n> 0} \gamma_n$. Define the relation $R \subseteq S_2 \times S_2$ to be
$$\{(\vec{v},\vec{v} * w) \mid \vec{v} \in S_2, w \in X_{\gamma(\vec{v}})\}$$
Note that there is, for every $\vec{v} \in S_2$, a bijection $i_{\vec{v}} : X_{\gamma(\vec{v})} \simeq R[\vec{v}]$ given by $w \mapsto \vec{v} * w$. Note also that, for each $\vec{v} \in S_2$, we have
$$ \gamma \circ i_{\vec{v}} = \xi_{\gamma(\vec{v})}$$

 With this in mind, we define the coalgebra structure $\sigma_2$ by setting
$$\sigma_2(\vec{v}) = T(i_{\vec{v}})(\alpha_{\gamma(\vec{v})})$$
Finally, set $s_2$ to be the unique singleton tuple with sole element $s_1$, and define the valuation $V_2$ by setting $V_2^\dagger(\vec{v}) = V_1^\dagger(\gamma(\vec{v}))$.

Clearly, $(S_2,R,\sigma_2,s_2,V_2)$ is a $\fun$-tree model. Denote the underlying $T$-model by $\mathbb{S}_2$.
We can then prove the following two claims:

\textbf{Claim 1}: The map $\gamma$ is a $T$-model homomorphism from $\mathbb{S}_2$ to $\mathbb{S}_1$.

\textbf{Claim 2}: $\mathbb{A}$ accepts $(S_2,R,\sigma_2,s_2,V_2)$ iff $\mathbb{A}^*$ accepts $\mathbb{S}_1$.

The proof of Claim 2 is left to the reader. We prove the first claim:

The map $\gamma$ clearly respects the truth values of all propositional atoms, and $\gamma(s_2) = s_1$. It suffices to show that $\gamma$ is a coalgebra morphism, i.e. that $\fun \gamma  (\sigma_2(\vec{v})) = \sigma_1 (\gamma(\vec{v}))$ for all $\vec{v}$. Pick any $\vec{v} \in S_2$. We have:
\begin{eqnarray}
\fun\gamma(\sigma_2(\vec{v})) & = & \fun\gamma \circ \fun(i_{\vec{v}})(\alpha_{\gamma(\vec{v})}) \\
& = & \fun(\gamma \circ i_{\vec{v}})(\alpha_{\gamma(\vec{v})}) \\
& = & \fun(\xi_{\gamma(\vec{v})})(\alpha_{\gamma(\vec{v})}) \\
& = & \sigma_1(\gamma(\vec{v}))
\end{eqnarray}
as required. \end{proof}

From this, a routine argument yields the following result.

\begin{theo}[Characterization Theorem 1]
\label{generalcharacterization}
Let $\Lambda$ be an expressively complete set of monotone predicate liftings 
for a set functor $\fun$, and assume that $\SOLa^{}(A)$ (for any set of
variables $A$) admits uniform translations. 
Then $\muMLLa$ is the bisimulation-invariant fragment of $\MSOLa$.
\end{theo}

The existence of uniform translations for the one-step language 
$\cite{vene:expr14}$ involves two components: a translation 
on the syntactic side and a uniform construction on the semantic side. 
However, as we shall now see, we can focus entirely on finding a suitable 
uniform construction for the one-step models; the syntactic translation will 
come for free.  

\begin{defi}
Let $\varphi$ be any formula of $\SOLa^{1}(A)$ of quantifier depth $\leq k$, 
and let $F$ be a uniform construction for $k$. 
Then, we define the generalized predicate lifting $\varphi^* : 
\cvp(-)^A \rightarrow \cvp \circ \fun$ by setting, for a given set $X$, 
$\alpha \in \fun X$ and $V : A \rightarrow \cvp(X)$: 
$$
\alpha \in\varphi^*_X(V) \text{ iff } 
(X_*,\alpha_*,V_{[h_{\alpha}]})\vDash_1 \varphi.
$$ 
\end{defi}

The following is obvious:

\begin{prop}
If $\varphi$ is a monotone formula then $\varphi^*$ is a monotone generalized predicate 
lifting.
\end{prop}

Note that, in order for $\SOLa^{1}(A)$ to admit a uniform translation, it 
suffices that there exists for any $k$ a uniform construction $F$ such that,
for every formula $\varphi$ of quantfier depth $\leq k$, the generalized 
lifting $\varphi^*$ is natural. An equivalent formulation of this condition
is the following.

\begin{prop}
Let $\varphi$ be any one-step formula in  $\SOLa^{1}(A)$  and let $F$ be a uniform construction. 
Then the lifting $\varphi^*$ is natural if, for any pair of sets $X,Y$,
any map $f : X \rightarrow Y$ and any valuation $V : A \rightarrow \cvp(Y)$, we 
have 
$$(\star) \quad (X_*,\alpha_*,V_{[f \circ h_\alpha]})\vDash_1 \varphi 
  \text{ iff } (Y_*, \beta_*, V_{[h_\beta]}) \vDash_1 \varphi
$$
provided that $\fun f(\alpha) = \beta$.
\end{prop} 
A uniform construction $F$ is said to be \textit{adequate} for $k$, and with 
respect to the liftings $\Lambda$, if the equivalence $(\star)$ holds for all 
(monotone) formulas in $\SOLa^{1}(A)$ of quantifier depth 
$\leq k$ (for any finite set of variables $A$). 
Since we could of course take the quantifier depth $k$ and the set of liftings
as extra inputs for the uniform construction, we shall simply say that the 
functor $\fun$ admits an adequate uniform construction if there is an adequate
uniform construction for $\fun$ with respect to every $k$ and every set of 
monotone liftings. 
If $\Lambda$ is an expressively complete set of liftings, this is equivalent to
requiring an adequate uniform construction with respect to $\Lambda$, for every
$k$.

The following theorem, from which we obtain Theorem~\ref{t:main2} by taking for
 $\La$  the set of all monotone liftings for $\fun$, summarizes
the results of this section.

\begin{theo}
\label{t:bisinv2}
Let $\La$ be any expressively complete set of monotone predicate 
liftings for the set functor $\fun$.
If $\fun$ admits an adequate uniform construction, then 
\[
\muMLLa \equiv \MSO_{\La}/{\simeq}.
\]
\end{theo}

\begin{example}
\label{ex:psf}
As a first application, the standard Janin-Walukiewicz characterization of the 
modal $\mu$-calculus can be seen as an instance of the result by taking
$\La = \{ \Box, \Diamond \}$ and $\fun = \psf$, recalling that $\MSO = \MSO_{\{\Diamond\}} \equiv \MSO_{\{\Box,\Diamond\}}$. 
The adequate uniform construction for $\psf$ is given as follows: consider a 
pair $(X,\alpha)$ with $\alpha \in \psf(X)$. 
We take this to $X_* = \alpha_* = \alpha \times \omega$, and we let $h_\alpha: 
\alpha \times \omega \to \alpha$ be the projection map.
\end{example}

It turns out that several other applications of this result can be obtained in
a particularly simple way. 
Say that a uniform construction $F$ is \textit{strongly adequate} if, for any 
mapping $f : X \to Y$ and any $\alpha \in \fun X$, $\beta \in \fun Y$ with 
$\fun f (\alpha) = \beta$, there is a bijection $g : X_{*}\to Y_{*}$ such that 
$\fun g (\alpha_{*}) = \beta_{*}$ and $f \circ h_\alpha = h_\beta \circ g$. 
Since it is easy to check that any strongly adequate uniform construction is 
adequate, we get:
\begin{coro}
If there is a strongly adequate uniform construction for $\fun$, then 
$\mu \mathtt{ML}_\fun\equiv\mathtt{MSO}_\fun {/} {\simeq} $. 
\end{coro} 

\begin{example}
\label{ex:bag}
As a first example,  consider the finitary multiset (``bags'') functor 
$\mathcal{B}$, which sends a set $X$ to the set of mappings $f : X \rightarrow
\omega$ such that the set $\{u \in X \mid f(u) = 0\}$ is cofinite.
The action on morphisms is given by letting, for $f \in \mathcal{B}X$ and 
$h : X \rightarrow Y$, the multiset $\mathcal{B}h(f) : Y \rightarrow \omega$ 
be defined by $w \mapsto \sum_{h(v) = w} f(v)$. 
Given a pair $X,\alpha$ where $\alpha : X \rightarrow \omega$ has finite support,
we define 
$$X_{*} = \bigcup \{\{u\} \times \alpha(u) \mid u \in X\}.
$$
Here, we identify each each $n \in \omega$ with the set $\{0,...,n-1\}$. 
The mapping $\alpha_{*} : X_{*} \rightarrow \omega$ is defined by setting 
$\alpha_{*}(w) = 1$ for all $w \in X_{*}$. 
The map $h_\alpha :  X_* \rightarrow X$ is defined by $(u,i) \mapsto u$. 
It is easy to check that the construction $F$ is strongly adequate, hence 
$\mu \mathtt{ML}_{\mathcal{B}}\equiv\mathtt{MSO}_{\mathcal{B}} {/} {\simeq}$.
\end{example}

As a final application, consider the set of all \textit{exponential polynomial 
functors}~\cite{jaco:intr12} defined by the ``grammar''
$$\mathsf{T} \isbnf \mathsf{C} 
   \divbnf \mathsf{Id} \divbnf \mathsf{T} \times \mathsf{T} 
   \divbnf \coprod_{i \in I} \mathsf{T}_i \divbnf \mathsf{T}(-)^\mathsf{C} 
$$
where $\mathsf{C}$ is any constant functor for some set $C$, and $\mathsf{Id}$ 
is the identity functor on $\mathbf{Set}$. 
These functors cover many important applications: streams, binary trees,
deterministic finite automata and deterministic labelled transition systems
are all examples of coalgebras for exponential polynomial functors, as is 
the socalled \textit{game functor} whose coalgebras provide the semantics for 
``Coalition Logic'' \cite{cirs:moda11}. 
\begin{prop}
\label{p:epf}
Every exponential polynomial functor admits a strongly adequate uniform 
translation.
\end{prop}

\begin{coro}
\label{c:epf}
For every exponential polynomial functor $\mathsf{T}$, we have 
$\mu \mathtt{ML}_{\fun}\equiv\mathtt{MSO}_{\fun} {/} {\simeq}$.
\end{coro}

The cases where we can find a strongly adequate uniform construction are the 
most straightforward applications of Theorem \ref{t:bisinv2} that we know of. 
The Janin-Walukiewicz theorem is a less direct application: there is no strongly 
adequate uniform construction for the powerset functor, but there is an adequate
uniform construction. 
In the next section, we shall study an example of a functor where there is no 
adequate uniform construction at all.

\section{The monotone neighborhood functor}
\label{sec:mon}

The final section of our paper concerns the monotone neighborhood functor 
$\mon$.
Our main result concerns a characterization of the fragment of $\MSO_{\mon}$
that is invariant under \textit{global} neighborhood bisimulations, to be 
introduced below. 
Our proof applies the method of section~\ref{sec:bisinv}, but not directly: 
we will first see that the functor $\mon$ itself does \emph{not} admit an
adequate uniform construction.

\subsection{No adequate uniform construction for $\mon$}

We first consider the negative result.

\begin{prop}
\label{p:no-adc-mon}
There is no adequate uniform construction for the monotone neighborhood 
functor $\mon$.
\end{prop}

\begin{proof}
To arrive at a contradiction assume that $F$ is adequate.
Fix some $a \in A$ and consider the formula $\varphi = \forall Z. (a \subseteq 
Z)$ expressing that $a$ has empty extension.

Let $Y$ be the set $\{u,v\}$ and let $\beta \in \mathcal{M}Y$ be the 
neighborhood structure $\{\{u\},\{u,v\}\}$.
Let $V$ be any valuation with $V(a) = \{v\}$.
First, we prove that
$(Y,\beta,V) \vDash_1 \varphi^*$:
to see this, consider the one-step model $(Y',\beta',V')$ where $\beta' = 
\{\{u\}\}$ and we recall that $Y' = \{u\}$, and where $V'$ is simply the 
restriction of $V$ to $Y'$. It is easy to show that
$( Y'_*, \beta'_*,V'_{[h_{\beta'}]}) \vDash_1 \varphi$, and hence
$(Y',\beta',V') \vDash_1 \varphi^*$.
Since the generalized predicate lifting $\varphi^*$ is natural by 
assumption and $\mon \iota_{Y',Y}(\beta')
= \beta$, we get $(Y,\beta,V) \vDash_1 \varphi^*$ as required.

With this in mind, let $X$ be the set $\{u^*,v^*,w^*\}$ and let $\alpha \in 
\mon X$ be the neighborhood structure
$$\{\{u^*,v^*\}, \{u^*,w^*\}, \{u^*,v^*,w^*\}\}$$
Define the map $f : X \rightarrow Y$ by setting $u^* \mapsto u$, $v^* \mapsto v$
and $w^* \mapsto u$. 
It can easily be checked that $\mon f (\alpha) = \beta$.
By naturality of the formula $\varphi^*$, it follows that
$(X,\alpha,V_{[f]}) \vDash_1 \varphi^*$.
Hence we must have
$$( X_*, \alpha_* ,V_{[f \circ h_{\alpha}]}) \vDash_1 \varphi$$
hence $V_{[f \circ h_\alpha]}(a) = \emptyset$. Since $v^* \in V_{[f]}(a)$, this
means that we have $v^* \notin h_\alpha[X_*]$. 
But since $\mon h_{\alpha}(\alpha_*) = \alpha$, this means $h_\alpha[X_*]$ 
must be a support for $\alpha$. But it is easy to show that $\alpha$ cannot 
have a support $S$ with $v^* \notin S$,
so we have now reached a contradiction showing that $F$ cannot be an 
adequate construction.
\end{proof}

\subsection{The  functor $\monstar$}

In this section, as a step towards our main characterization result, we shall 
consider the language $\muML_{\monstar}$, where the functor 
$\monstar$ is a slight variation of the monotone neighborhood functor 
$\mon$. 
The functor $\monstar$ is obtained as the subfunctor of $\mon 
\times \psf$ given by
$$X \mapsto \{(\alpha,Y) \in \mon X \times \psf X \mid Y 
  \text{ supports } \alpha\}$$                           
This is indeed a subfunctor of $\mon \times \psf$, because given a map 
$h : X \rightarrow Y$, if $Z$ is a support for $\alpha \in \mon X$, then
$h[Z]$ is a support for $\mon h(\alpha)$. Given $\alpha \in 
\monstar X$, we will write $\alpha = (N_\alpha,S_\alpha)$. 

\begin{defi}
For the functor $\monstar$ we define the unary predicate liftings $\Box$
and $E$ by
\begin{align*}
\Box_X(Z) & \isdef \{\alpha \in \monstar X \mid Z \in N_\alpha\}
\\
E_X (Z)   & \isdef 
  \{\alpha \in \monstar X \mid Z \cap S_\alpha \neq \emptyset\},
\end{align*}
and we let $\Diamond$ be the dual of  $\Box$ and let $E^d$ be the dual of $E$. 
The set of liftings $\{\Box,\Diamond, E,E^d\}$ is denoted as $\Theta$.
\end{defi}

The set $\Theta$ is an expressively complete set of liftings for $\monstar$. We shall omit the proof of this fact here, and merely state it as the following proposition:

\begin{prop}
\label{p:lyndon}
Every monotone natural predicate lifting 
$\lambda : \cvp(-)^A \rightarrow \cvp \circ \monstar$ is equivalent to a 
formula in $\ML_{\Theta}(A)$.
\end{prop}

The main technical result of this section states the existence, for all $k$,
of a uniform construction $F$ that is adequate for $k$ and with respect to 
the set of liftings $\{\Box,E\}$. 

\begin{defi}
\label{d:F-monstar}
Fix a natural number $k$. 
Given a set $X$, and object $\alpha  \in \monstar X$, put
$$
X_* \isdef 
\{(u,i,Z,j) \in (X \times 2^{k} \times \psf(S_\alpha) \times \omega) 
  \mid u \in Z \},
$$
and let $\pi_X$ be the projection map from $X_*$ to $X$.
Define $\alpha_* = (N_{\alpha_*},S_{\alpha_*}) \in \monstar(X_*)$ 
by setting $S_{\alpha_*} = X_*$, and set $Z \in N_{\alpha_*}$ for $Z 
\subseteq S_{\alpha_*}$ iff $\left\lceil Y,j\right\rceil \subseteq Z$ for some
$Y \in \alpha$, $Y \subseteq S_\alpha$ and some $j < \omega$, where
$$\left\lceil Y,j \right\rceil := \{(u,i,Y,j) \mid u \in Y, i < 2^k 
\}.
$$
The sets of the form $\left\lceil Z,j\right\rceil$ will be called the 
\textit{basic members} of $N_{\alpha^*}$.
\end{defi}

The main goal of this section is to prove the following:
\begin{prop}
\label{p:ad-monstar}
The construction given in
Definition~\ref{d:F-monstar} is an adequate, uniform construction for $k$.
\end{prop}

It is easy to check that, for all sets $X$ and $\alpha \in \monstar X$, we have $\monstar\pi_X(\alpha_F) = \alpha$.

Our main goal in this section is to prove the following result, from which Proposition~\ref{p:ad-monstar} now follows:
\begin{lemma}
\label{monuniform}
Let $X,Y$ be any sets, $\alpha \in \monstar X$, $\beta \in \monstar Y$ 
and $V : A \rightarrow Q(Y)$. Suppose that we have a map $h : X \rightarrow Y$ such that $\monstar h(\alpha) = \beta$. Then we have
$$(X_*,\alpha_*,V_{[h \circ \pi_X]}) \equiv^k (Y_*,\beta_*,V_{ [\pi_Y]})$$
\end{lemma}

Here, and throughout this section, we write $(X,\alpha,V) \equiv^k (Y,\beta,U)$ to say that two one-step models satisfy the same formulas of $\MSO_{\{\Box,E\}}^{1}(A)$ with at most $n$ nested quantifiers. Let us keep the data $X,Y,\alpha,\beta,V$ and $h$ fixed throughout the proof, and assume that $\monstar h(\alpha) = \beta$. We will also assume, from now on, that $N_\alpha$ and $N_\beta$ are both non-empty sets: if one of them is empty then both of them are, and in this case the lemma can be proved essentially using an easier version of the argument we use below.  

\begin{defi}
Given a finite set of variables $A$, a propositional $A$-\textit{type} $\tau$ is a subset of $A$. Given a set $X$ and a valuation $V : A \rightarrow Q(X)$, the propositional $A$-type of $v \in X$ is defined to be $V^\dagger(v) = \{a \in A \mid v \in V(a)\}$.
\end{defi}

\begin{defi}
Given a basic member $\left\lceil Z,j \right\rceil$  either in $N_{\alpha_*}$ or in $N_{\beta_*}$, a valuation 
$V : A \rightarrow \cvp(X_*)$ or $V : A \rightarrow \cvp(Y_*)$, and a natural number $m$, the \textit{$m$-signature} of $\left\lceil Z,j \right\rceil $ over variables $A$ and relative to the valuation $V$ is the mapping $\sigma : \psf(A) \rightarrow \{0,...,m\} $ defined by setting $\sigma(t)$ to be $n < m$ if $\left\lceil Z,j \right\rceil$ contains exactly $n$ elements of type $t$ under the valuation $V$, or $\sigma(t) = m$ if $\left\lceil Z,j \right\rceil$ contains \textit{at least} $m$ elements of type $t$. 
\end{defi}
 
\begin{defi}
 Let $B$ be any set of variables containing $A$, and let $V_1 : B \rightarrow \cvp(X_*)$ and $V_2 : B \rightarrow \cvp(Y_*)$. Then for any natural number $n$ we write  
$$(X_*,\alpha_*,V_1) \approx^{n} (Y_*,\beta_*,V_2) $$
and say that these one-step models \textit{match up to depth $n$}, if: for every $n$-signature $\sigma$ over variables $B$, either the number of basic elements of signature $\sigma$ in $N_{\alpha_*}$ and $N_{\beta_*}$ respectively are both finite and the same, or both infinite.
\end{defi}

\begin{lemma}
\label{kstep}
 $(X_*,\alpha_*,V_{[h \circ \pi_X]}) \approx^{2^k} (Y_*,\beta_*,V_{ [\pi_Y]})$.
\end{lemma}

\begin{proof}
First note that, for any $2^k$-signature $\sigma$, $N_{}$ either contains no basic elements of signature $\sigma$, or infinitely many: if there is some basic element $\left\lceil Z,j \right\rceil$ if signature $\sigma$, then for any $i \neq j$, the basic element $\left\lceil Z,i \right\rceil$ has the same $2^k$-signature as $\left\lceil Z,j \right\rceil$ with respect to the valuation $V_{h \circ \pi_X}$. The same holds for $N_{\beta_*}$ with respect to the valuation $V_{ \pi_Y}$.

Thus, it suffices to show that $N_{\alpha_*}$ contains a basic element of signature $\sigma$ w.r.t. $V_{[h \circ \pi_X]}$ iff $N_{\beta_*}$ contains a basic element of signature $\sigma$ w.r.t $V_{[\pi_Y]}$. Suppose that $N_{\beta_*}$ contains a basic element $\left\lceil Z,j\right\rceil$ of signature $\sigma$, where $Z \in N_\beta$ and $Z \subseteq S_\beta$. Since $N_\beta = \mathcal{M}h(N_\alpha)$, we have $h^{-1}[Z] \in N_\alpha$ and so $h^{-1}[Z] \cap S_\alpha \in N_\alpha$. Hence we get
$$\left\lceil h^{-1}[Z] \cap S_\alpha, 0\right\rceil \in N_{\alpha_*}$$
It is easy to see that any basic element of $N_{\alpha_*}$ contains either $0$ or at least $2^k$ members of a propositional type $t$ w.r.t. $V_{[h \circ \pi_X]}$, and the same is true of any basic element of $N_{\beta_*}$ w.r.t. $V_{[\pi_Y]}$. Hence, to show that $\left\lceil h^{-1}[Z] \cap S_\alpha, 0\right\rceil $ has the same $2^k$-signature as $\left\lceil Z,j\right\rceil$, it suffices to show that these basic elements realize the same propositional types over $B$. So suppose $\left\lceil h^{-1}[Z] \cap S_\alpha, 0\right\rceil $ contains some element $v$ of type $t$. Then $\pi_X(v) \in h^{-1}[Z] \cap S_\alpha$, so $h \circ \pi_X(v) \in Z$. This means that
$$(h \circ \pi_X(v),0,Z,j) \in \left\lceil Z,j\right\rceil $$
and this will have the same propositional type as $v$. Conversely, suppose that $\left\lceil Z,j\right\rceil $ contains some element $w$ of type $t$. Then $\pi_Y(w) \in Z$. Since $Z \subseteq S_\beta$, and $S_\beta = h[S_\alpha]$, there exists some $w' \in h^{-1}[Z] \cap S_\alpha$ with $h(w') = \pi_Y(w)$. We get
$$(w',0,h^{-1}[Z] \cap S_\alpha,0) \in \left\lceil h^{-1}[Z] \cap S_\alpha, 0\right\rceil $$
and this will have the type $t$.

Now, suppose that $N_{\alpha_*}$ contains a basic element   $\left\lceil Z,j\right\rceil$ of $2^k$-signature $\sigma$, so that $Z \in N_\alpha$ and $Z \subseteq S_\alpha$. Then $h[Z] \in N_\beta$, and furthermore $h[Z] \subseteq S_\beta = h[S_\alpha]$. Again, it suffices to check that $\left\lceil Z,j\right\rceil$ and $\left\lceil h[Z],0\right\rceil$ realize the same propositional types. Given $v\in \left\lceil Z,j\right\rceil$, we have $h \circ\pi_X(v) \in h[Z]$ and so
$$ (h \circ \pi_X(v),0,h[Z],0) \in \left\lceil h[Z],0\right\rceil$$
and this will have the same propositional type as $v$. Conversely, if $w \in \left\lceil h[Z],0\right\rceil$ then $\pi_Y(w) \in h[Z]$, so there is $w' \in Z$ with $h(w') = \pi_Y(w)$. We get
$$(w',0,Z,j) \in \left\lceil Z,j\right\rceil$$
and this has the same propositional type as $w$.  
\end{proof}

We are going to show, by induction on a natural number $m \leq k$, that if two one-step models  of the form $(X_*,\alpha_*,V_1)$ and $(Y_*,\beta_*,V_2)$ match up to depth $2^m$, then they satisfy the same formulas of quantifier depth $m$. For the basis case of $2^0 = 1$, we need the following result:
\begin{lemma}
\label{atomicstep}
Let $B$ be a set of variables containing $A$, and let $V_1 : B \rightarrow \cvp(X_*)$ and $V_2 : B \rightarrow \cvp(F^{\beta}(Y))$ be valuations such that
$$(X_*,\alpha_*,V_1) \approx^{1} (Y_*,\beta_*,V_2) $$
Then these two one-step models satisfy the same atomic formulas of the one-step
language $\MSO^{1}_{\{\Box,E\}}$.
\end{lemma}

\begin{proof}
Suppose first that
$$X_*,\alpha_*,V_1)\vDash_1 p \subseteq q$$
where $p,q \in B$. Suppose that $V_2(p) \nsubseteq V_2(q)$. Then there is some $(u,i,Z,j) \in Y_*$ such that $(u,i,Z,j) \in V_2(p)\setminus V_2(q)$. We have $(u,i,Z,j) \in \left\lceil Z,j\right\rceil$, and so there must be some basic element $\left\lceil Z',j'\right\rceil$ of $N_{\alpha_*}$ of the same $1$-signature over variables $B$ as $\left\lceil Z,j\right\rceil$. It follows that there is some element of $\left\lceil Z',j'\right\rceil$ of the same propositional type as $(u,i,Z,j)$, and then we cannot have $V_1(p) \subseteq V_1(q)$. The converse direction is proved in the same manner.

Now, suppose that 
$$(X_*,\alpha_*,V_1)\vDash_1 \Box p$$
Then $V_1(p) \in N_{\alpha_*}$, so there is some basic element $\left\lceil Z,j\right\rceil \in N_{\alpha_*}$ with $\left\lceil Z,j\right\rceil \subseteq V_1(p)$. There must be some basic $\left\lceil Z',j'\right\rceil \in N_{\beta_*}$ of the same $1$-signature over $B$ as $\left\lceil Z,j\right\rceil$, and clearly it follows that $\left\lceil Z',j'\right\rceil \subseteq V_2(p)$ and so $V_2(p) \in N_{\beta_*}$ as required. The converse direction is proved in the same way.

Finally, suppose that 
$$(X_*,\alpha_*,V_1)\vDash_1 E p$$
Then there is some $(u,i,Z,j) \in S_{\alpha_*}$ with $(u,i,Z,j) \in V_1(p)$. We have $(u,i,Z,j) \in \left\lceil Z,j\right\rceil$, and  there must be some basic element  $\left\lceil Z',j'\right\rceil$ of $N_{\beta_*}$ with the same $1$-signature as $\left\lceil Z,j\right\rceil$. Hence, $\left\lceil Z',j'\right\rceil$ contains some element $(u',i',Z',j')$ of the same propositional type as $(u,i,Z,j)$, and it follows that
$$(Y_*,\beta_*,V_2)\vDash_1 E p$$
as required. The converse direction is proved in the same way.
\end{proof}
To clinch the proof of Proposition \ref{p:ad-monstar}, we now only need the following lemma:
\begin{lemma}
\label{mainlemmamonstar}
Let $B$ be a finite set of variables containing $A$, let $0 < m \leq k$ and let $V_1 : B \rightarrow Q(X_*)$ and $V_2 : B \rightarrow \cvp (Y_*)$ be valuations such that
$$(X_*,\alpha_*,V_1) \approx^{2^m} (Y_*,\beta_*,V_2)$$
Let $q$ be any fresh variable. Then for any valuation
$V_1'$ 
extending $V_1$ with some value for $q$, there exists a valuation
$V_2' $ 
extending $V_2$, such that
$$(X_*,\alpha_*,V_1') \approx^{2^{(m-1)}} (Y_*,\beta_*,V_2')$$
and vice versa. 
\end{lemma}

\begin{proof}
We only prove one direction since the other direction can be proved by a symmetric argument. Let $V_1'$ be given. By the hypothesis, for any $2^m$-signature $\sigma$ over the variables $B$, either the number of basic elements of signature $\sigma$ in $N_{\alpha_*}$ and $N_{\beta_*}$ relative to $V_1$ and $V_2$ are both finite and the same, or both infinite. Let $\sigma_1,...,\sigma_k$ be a list of all the distinct $2^m$-signatures over $B$ such that the set of basic elements of $N_{\alpha_*}$ and $N_{\beta_*}$ of signature $\sigma_i$, with $1 \leq i \leq k$, is non-empty but finite, and let $\sigma_{k+1},...,\sigma_{l}$ be a list of all the $2^m$-signatures such that, for $k + 1 \leq i \leq l$, there are infinitely many basic elements of $N_{\alpha_*}$ and of $N_{\beta_*}$ of signature $\sigma_i$. Then, for each $i \in \{1,...,l\}$, let $\alpha_*[\sigma_i]$ denote the set of basic elements in $N_{\alpha_*}$ of signature $\sigma_i$, and similarly let $\beta_*[\sigma_i]$ denote the set of basic elements of $N_{\beta_*}$ of signature $\sigma_i$. Then $\alpha_*[\sigma_1],...,\alpha_*[\sigma_l]$ is a partition of the set of basic elements of $N_{\alpha_*}$ into non-empty cells, and similarly $\beta_*[\sigma_1],...,\beta_*[\sigma_l]$ is a partition of the set of basic elements of $N_{\beta_*}$.

Given the extended valuation $V_1'$ in $X_*$ defined on variables $B \cup \{q\}$, we similarly let $\tau_1,...,\tau_{k^*}$ be a list of all the $2^{m - 1}$-signatures over $B \cup \{q\}$ such that, for $1 \leq i \leq k^*$, the set of basic elements of $N_{\alpha_*}$ of $2^{m - 1}$-signature $\tau_i$ is non-empty but finite. We let $\tau_{k^* + 1},...,\tau_{l^*}$ be a list of all the $2^{m - 1}$-signatures over $B \cup \{q\}$ such that, for each $i$ with $k^* + 1 \leq i \leq l^*$, the set of basic elements of $N_{\alpha_*}$ of $2^{m - 1}$-signature $\tau_i$ is infinite. Let $\alpha_*[\tau_i]$ denote the set of basic elements of $N_{\alpha_*}$ of $2^{m - 1}$-signature $\tau_i$, so that the collection $\alpha_*[\tau_1],...,\alpha_*[\tau_{l^*}]$ constitutes a second partition of the set of basic elements of $N_{\alpha*}$. It will be useful to introduce the abbreviation $D_1$ for the finite set $\alpha_*[\sigma_1] \cup...\cup \alpha_*[\sigma_k]$, and the abbreviation $D_2$ for the finite set $\alpha_*[\tau_1] \cup...\cup \alpha_*[\tau_{k^*}]$. 

For each $i$ with $1 \leq i \leq k$, there is a bijection between the set $\alpha_*[\sigma_i]$ and $\beta_*[\sigma_i]$, and we can paste all these bijections together into a bijective map 
$$f : \alpha_*[\sigma_1] \cup...\cup \alpha_*[\sigma_k] \rightarrow \beta_*[\sigma_1] \cup ... \cup \beta_*[\sigma_k]$$ 
Since every basic element of $N_{\alpha_*}$ not in $D_1$ belongs to a $2^m$-signature of which there are infinitely many basic elements in $\beta_*$, and since $D_1 \cup D_2$ is finite, it is easy to see that we can extend the map $f$ to a map $g$ which is an injection from the set $D_1 \cup D_2$
into the set of basic elements of $N_{\beta_*}$, such that for each basic element $\left\lceil Z,j\right\rceil$ in $D_1 \cup D_2$, $\left\lceil Z,j\right\rceil$ and $g(\left\lceil Z,j\right\rceil)$ have the same $2^m$-signature over $B$, and such that $g\!\upharpoonright\!D_1 = f$. Each basic element of $N_{\beta_*}$ not in the image of $g$ must then be of one of the $2^m$-signatures $\sigma_{k + 1},...,\sigma_{l}$, and so we can partition the set of basic elements of $N_{\beta_*}$ outside the image of $g$ into the cells
$\beta_*[\sigma_{k + 1}]\setminus g[D_2],...,\beta_*[\sigma_l]\setminus g[D_2]$.
For each $i$ with $k + 1 \leq i \leq l$, let $\gamma^i_{1},...,\gamma^i_{r}$ list all infinite sets of the form 
$\alpha_*[\sigma_i] \cap \alpha_*[\tau_j]$ for $k^* + 1 \leq j \leq l^*$.
The list $\gamma^i_{1},...,\gamma^i_{r}$ must be non-empty, and so since the set $\beta_*[\sigma_{i}]\setminus g[D_2]$ is also infinite, we may 
partition it into $r$ many infinite cells and list these as $\delta^i_{1},...,\delta^i_{r}$.
Now, for each basic element $\left\lceil Z,j\right\rceil$ of $\beta_*$, we 
define a map 
$W_{\left\lceil Z,j\right\rceil}$ from  $B \cup \{q\}$ to 
  $\psf(\left\lceil Z,j\right\rceil)$
by a case distinction as follows:

\emph{Case 1:} $\left\lceil Z,j\right\rceil = g(\left\lceil Z',j'\right\rceil)$ 
for some $\left\lceil Z',j'\right\rceil \in D_1 \cup D_2$. 
Then $\left\lceil Z,j\right\rceil$ and $\left\lceil Z',j'\right\rceil$ have the
same $2^m$-signature over $B$. Using this fact we define the valuation 
$W_{\left\lceil Z,j\right\rceil}$ so that, for each $p \in B$, we have 
$W_{\left\lceil Z,j\right\rceil}(p) = V_2(p) \cap \left\lceil Z,j\right\rceil$,
and so that $\left\lceil Z',j'\right\rceil$ and $\left\lceil Z,j\right\rceil$ 
have the same $2^{m - 1}$-signature over $B \cup \{q\}$ with respect to the
valuations $V_1'$ and $W_{\left\lceil Z,j\right\rceil}$.

 We show how to assign the value of the variable $q$: for each propositional type $t$ over $B \cup \{q\}$, there are three different possible cases to consider. If $\left\lceil Z',j'\right\rceil$ has $l < 2^{m-1}$ elements of type $t \cup \{q\}$ over $B \cup \{q\}$, then pick $l$ many elements of $\left\lceil Z,j\right\rceil$ of type $t$ and mark them by $q$. This is possible since $l < 2^{m-1} \leq 2^m$ and $\left\lceil Z',j'\right\rceil$ and $\left\lceil Z,j\right\rceil$ have the same $2^m$-signature. If there are $l < 2^{m-1}$ elements of $\left\lceil Z',j'\right\rceil$ of type $t$ over $B \cup \{q\}$, then pick $l $ elements of $\left\lceil Z,j\right\rceil$ of type $t$ over $B$, and mark all the other elements of $\left\lceil Z,j\right\rceil$ of type $t$ by $q$. Finally, if there are at least $2^{m-1}$ elements of $\left\lceil Z',j'\right\rceil$ of type $t \cup \{q\}$ over $B \cup \{q\}$ and at least $2^{m - 1}$ elements of $\left\lceil Z',j'\right\rceil$ of type $t$ over $B \cup \{q\}$, then all in all there must be at least $2^m$ elements of $\left\lceil Z',j'\right\rceil$ of type $t$ over $B$, and so there must be at least $2^m$ elements of $\left\lceil Z,j\right\rceil$ of type $t$ over $B$. Pick $2^{m-1}$ of these and mark them by $q$. Finally, let $W_{\left\lceil Z,j\right\rceil}(q)$ be the set of elements of $\left\lceil Z,j\right\rceil$ marked by $q$.

\emph{Case 2:} $\left\lceil Z,j\right\rceil$ is not in the image of $g$. Then there must be some $i \in \{k^* + 1,...,l^*\}$ such that $\left\lceil Z,j\right\rceil \in \beta_*[\sigma_{k + 1}]\setminus g[D_2]$, and this set is partitioned into $\delta^i_1,...,\delta^i_r$. Let $\left\lceil Z,j\right\rceil \in \delta^i_j$, and pick some arbitary element $\left\lceil Z',j'\right\rceil$ of the set $\gamma^i_j$. Then $\left\lceil Z',j'\right\rceil$ and $\left\lceil Z,j\right\rceil$ have the same $2^m$-signature over $B$ and we can proceed as in Case 1.

We define the valuation $V_2'$ by setting $V_2'(q)$ to be the union of the sets $W_{\left\lceil Z,j\right\rceil}(q)$ for $\left\lceil Z,j\right\rceil$ a basic element in $N_{\beta_*}$. It is now fairly straightforward to check that
$$(X_*,\alpha_*,V_1') \approx^{2^{(m-1)}} (Y_*,\beta_*,V_2')$$  
as required. 

First, suppose there are infinitely many basic elements of $N_{\alpha_*}$ of some $2^{m-1}$ signature $\tau_j$, meaning that $k^* \leq j \leq l^*$. Then since the set $\alpha_*[\tau_j]$ is infinite, $D_1$ is finite and $\alpha_*[\tau_j]$ is equal to
$$ (D_1 \cap \alpha_*[\tau_j]) \cup (\alpha_*[\sigma_{k +1}] \cap \alpha_*[\tau_{j}]) \cup ... \cup (\alpha_*[\sigma_l] \cap \alpha_*[\tau_{j}])$$
 there must be some $i \in \{k+1,...,l\}$ such that the set $\alpha_*[\sigma_i] \cap \alpha_*[\tau_j]$ is infinite. This means that $\alpha_*[\sigma_i] \cap \alpha_*[\tau_j]$ appears in the list $\gamma^i_1,...,\gamma^i_r$, and so we see that all elements of some member of the list $\delta^i_1,...,\delta^i_r$ will have the $2^{m-1}$-signature $\tau_j$. Since each member of this list is infinite, we see that there must be infinitely many basic elements of $N_{\beta_*}$ of signature $\tau_j$. 

Conversely, suppose there are infinitely many basic elements of $N_{\beta_*}$ of $2^{m-1}$-signature $\tau_j$ over $B \cup \{q\}$. Then since the image of $g$ is finite, some of these elements must be outside the image of $g$, which means that for some $i \in \{k + 1,...,l\}$, some member of the list $\delta^i_1,...,\delta^i_r$ will consist of elements of signature $\tau_j$. This means that some member of the list $\gamma_1^i,...,\gamma_r^i$ will consist of elements of signature $\tau_j$, and since each member of this list is infinite we see that $N_{\alpha_*}$ has infinitely many basic elements of $2^{m-1}$-signature $\tau_j$ over $B \cup \{q\}$.

Finally, suppose that there are finitely many basic elements of $N_{\alpha_*}$ and $N_{\beta_*}$ of $2^{m-1}$-signature $\tau_j$. We check that the mapping $g$ restricts to a bijection between the basic elements of $N_{\alpha_*}$ and $N_{\beta_*}$ of this signature. First, $g$ is injective and maps basic elements of $N_{\alpha_*}$ of signature $\tau_j$ to basic elements of $N_{\beta_*}$ of signature $\tau_j$. It only remains to show that (the restriction of) $g$ is surjective, i.e. each basic element $\left\lceil Z,r\right\rceil$ of signature $\tau_j$ is equal to $g(\left\lceil Z',r'\right\rceil)$ for some $\left\lceil Z',r'\right\rceil$. But suppose $\left\lceil Z,r\right\rceil$ is not in the image of $g$; then it is in one of the members of the list $\delta_1^i,...,\delta_r^i$ for some $i$, and since each of these members is an infinite set of basic elements of the same signature, we see that there are infinitely many basic elements of $N_{\beta_*}$ of signature $\tau_j$, contrary to our assumption. Hence, the proof is done. 
\end{proof}
Lemma \ref{monuniform} can now be deduced by combining the last three lemmas, by a straightforward argument using Ehrenfeucht-Fraïssé games for the one-step language.

\subsection{Global neighborhood bisimulations}

Since the set of liftings $\{\Box,\Diamond\}$ can be shown to be expressively complete for $\mon$, and since $\Diamond$ is just the dual of $\Box$,
the monadic second order language $\MSO_\mon$ is equivalent to 
the logic $\mathtt{MMSO}$ which has its syntax given by
$$
\varphi \isbnf
   \sr(p) \mid p \subseteq q \mid \Box(p,q) \mid \exists p. \varphi 
   \mid \neg \varphi \mid \varphi \vee \varphi.
$$
The semantics of an atomic formula $\Box(p,q)$ in a neighborhood model 
$\mathbb{S}$ is given, concretely, by the clause: 
$(S,\sigma,V,u)\vDash \Box(p,q)$ if, for all $v \in V(p)$, there is 
$Z \in \sigma(v)$ such that $Z \subseteq V(q)$. 

At the present time, we do not know how to characterize the fragment of the 
language $\mathtt{MMSO}$ that is invariant for arbitrary neighborhood
bisimulations. 
However, the situation changes if we consider \textit{global} bisimulations 
between neighborhood models.

\begin{defi}
A \textit{global neighborhood bisimulation} between $\mon$-models 
$\mathbb{S}_1$ and $\mathbb{S}_2$ is a neighborhood bisimulation $R$ that 
satisfies the conditions:
\begin{description}
\item[Forth] For every $u \in S_1$ there is some $v\in S_2$ with $u R v$
\item[Back] For every $v \in S_2$ there is some $u \in S_1$ with $u R v$
\end{description}
\end{defi}
We now ask: what is the fragment of $\mathtt{MMSO}$ that is invariant
for global neighborhood bisimulations? Since global bisimulations are the natural
equivalence relation for modal logic with the 
\textit{global modalities}, the most reasonable candidate would
be: the monotone modal $\mu$-calculus extended with the global modalities. 
To be precise, let the \textit{monotone modal $\mu$-calculus with global 
modalities}, denoted $\mu \mathtt{MML}_g$, be the language defined by
the grammar:
\begin{align*}
\varphi \isbnf p & \mid \neg p \mid \bot \mid  \top \mid \Box \varphi 
   \mid \Diamond\varphi \mid [\forall] \varphi \mid [\exists] \varphi 
\\ &   \mid \varphi \vee \varphi \mid  \varphi \wedge \varphi \mid\mu p. \varphi 
   \mid \nu p.\varphi
\end{align*}
where the formula $\varphi$ in $\mu p. \varphi$ and $\nu p . \varphi$ must be 
positive in the variable $p$.  
The new operators $[\forall]$ and $[\exists]$ are the global universal and 
existential modalities, with their standard semantics: 
$(\mathbb{S},u) \vDash [\forall] \varphi$ if $(\mathbb{S},v) \vDash \varphi$ 
for all $v \in S$, and $(\mathbb{S},u) \vDash [\exists] \varphi$ if
$(\mathbb{S},v) \vDash \varphi$ for some $v \in S$. 

Given an $\monstar$-model $\mathbb{S}$, let $\mathbb{S}_\mon$ be the underlying $\mon$-model. Conversely, given an $\mon$-model $\mathbb{S} = (S,\sigma,V)$, define the $\monstar$-model $\mathbb{S}^G = (S,\sigma^G,V)$ by setting $\sigma^G(s) = (\sigma(s),S)$.
The main result of this section is the following.
\begin{theo}
\label{janinwalukiewicztheorem}
A formula in $\mathtt{MMSO}$ is invariant for global neighborhood 
bisimulations if, and only if, it is equivalent to a formula of the logic  
$\mu \mathtt{MML}_g$.
\end{theo}

\begin{proof}
Clearly  $\mu \mathtt{MML}_g$ translates into 
$\mathtt{MMSO}$ and is invariant for global bisimulations. 

Conversely, suppose $\varphi \in \mathtt{MMSO} $ is invariant for
global neighborhood bisimulations. 
First observe that $\varphi$ can be regarded as a formula in 
$\MSO_{\monstar}$ as well. 
More precisely, there is a formula $\varphi^*  \in \MSO_{\monstar}$
such that
\begin{equation}
\label{eq:y0}
(\mathbb{T},t) \vDash \varphi^* \text{ iff } 
(\mathbb{T}_\mon,t) \vDash \varphi
\end{equation}
for any $\monstar$-model $(\mathbb{T},t)$.
By Corollary~\ref{c:automatachar} there is a second-order $\{\Box,E\}$-automaton 
$\bbA_{\phi}$ such that 
\begin{equation}
\label{eq:y1}
\bbA_{\phi} \equiv \phi^{*} \text{ (on all $\monstar$-tree models)}.
\end{equation}

Now we use the existence of an adequate, uniform construction for $\monstar$
(Proposition~\ref{p:ad-monstar}).
Let $\bbA_{\phi}^{t}$ be the corresponding modal $\Lambda$-automaton given
by Proposition~\ref{p:unr}, where $\Lambda$ is the collection of all monotone, 
natural predicate liftings for $\monstar$.
By Proposition~\ref{p:lyndon} we may in fact assume that $\bbA_{\phi}^{t}$ is
a $\Theta$-automaton, where $\Theta = \{ \Box, \Diamond,  E, E^{d} \}$.
Let $\psi = \psi_{\bbA_{\phi}^{t}}$ be the corresponding formula in
$\muML_{\Theta }$. 
We claim that, for any pointed neighborhood model $(\bbS,s )$ we have
\begin{equation}
\label{eq:y2}
\mathbb{S},s \vDash \varphi \text{ iff } 
\mathbb{S}^{G},s \vDash \psi.
\end{equation}

To prove this, consider the $\monstar$-tree model $(\bbT,R,r)$ given by
Proposition~\ref{p:unr}, applied to the pointed $\monstar$-model 
$(\bbS^{G},s)$.
Then there is a surjective $\monstar$-coalgebra morphism $f: (\bbT,r) \to 
(\bbS^{G},s)$, and so in particular, $f$ is the graph of a global 
neighborhood bisimulation between $\bbT_{\mon}$ and $\bbS$ relating $r$ to $s$.
Gathering some facts we obtain the following chain of equivalences:
\begin{align*}
\bbS,s \vDash \phi & \text{ iff } \bbT_{\mon},r \vDash \phi
   & \text{(assumption on $\phi$)}
\\ & \text{ iff } \bbT,r \vDash \phi^{*}
   & \text{\eqref{eq:y0}}
\\ & \text{ iff } \bbT,R,r \vDash \bbA_{\phi}
   & \text{\eqref{eq:y1}}
\\ & \text{ iff } \bbS^{G},s \vDash \bbA_{\phi}^{t}
   & \text{(Proposition~\ref{p:unr})}
\\ & \text{ iff } \bbS^{G},s \vDash \psi
   & \text{(assumption on $\psi$)}
\end{align*}
which proves \eqref{eq:y2} indeed.

Finally, let $\psi^{\forall} \in \mu\mathtt{MML}_g$ be the formula
we obtain from $\psi$ by replacing every occurrence of $E$ with $[\exists]$ and
every occurrence of $E^d$ with $[\forall]$.
It is a routine check to verify that 
\begin{equation}
\label{eq:y3}
\mathbb{S}^{G},s \vDash \psi \text{ iff } 
\mathbb{S},s \vDash \psi^{\forall}.
\end{equation}
But then the equivalence of $\phi \in \mathtt{MMSO}$ and
$\psi^{\forall} \in \mu \mathtt{MML}_g$ is immediate from
\eqref{eq:y2} and \eqref{eq:y3}.
\end{proof}

\bibliographystyle{abbrv}

\bibliography{mu,automata,nabla,extra}

\end{document}